\documentclass[11pt]{article}

\usepackage{fullpage}
\usepackage[bookmarks]{hyperref}
\usepackage{amssymb}
\usepackage{amsmath}
\usepackage{amsthm}
\usepackage{graphicx,color,colordvi}
\usepackage{bbm}
\usepackage{cleveref}
\usepackage{stmaryrd}
\usepackage[utf8]{inputenc}
\usepackage{fullpage}
\usepackage[blocks]{authblk}
\usepackage{mathtools}
\usepackage[all]{xy}
\usepackage{pgfplots}

\def\openone{\leavevmode\hbox{\small1\kern-3.8pt\normalsize1}}

\def\CC{\mathbb{C}}
\def\RR{\mathbb{R}}

\def\NN{\mathbb{N}}
\def\LL{\mathbb{L}}

\def\11{\mathbf{1}}
\def\LL{\mathcal{L}}
\def\PP{\mathbb{P}}
\def\EE{\mathbb{E}}

\usepackage{pst-node}
\usepackage{tikz-cd} 

\newtheorem{theorem}{Theorem}
\newtheorem{lemma}{Lemma}
\newtheorem{proposition}{Proposition}

\theoremstyle{definition}

\def\reff#1{(\ref{#1})}
\def\eps{\varepsilon}

\newcommand{\supp}{\mathop{\rm supp}\nolimits}

\newcommand{\tr}{\mathop{\rm Tr}\nolimits}

\newcommand{\spec}{{\rm sp}}

\newcommand{\cB}{{\cal B}}
\newcommand{\cD}{{\cal D}}
\newcommand{\cE}{{\cal E}}

\newcommand{\cH}{{\cal H}}

\newcommand{\cJ}{{\cal J}}

\newcommand{\cP}{\mathcal{P}}

\newcommand{\cL}{{\cal L}}

\def\e{\mathrm{e}}

\usepackage{graphicx}
\usepackage{setspace}
\usepackage{verbatim}
\usepackage{subfig}

\theoremstyle{definition}

\theoremstyle{remark}
\newtheorem{remark}{Remark}

\numberwithin{equation}{section}

\DeclareRobustCommand\openone{\leavevmode\hbox{\small1\normalsize\kern-.33em1}}

\newcommand{\id}{\rm{id}}
\newcommand{\be}{\begin{equation}}
	\newcommand{\ee}{\end{equation}}
\newcommand{\bea}{\begin{eqnarray}}
	\newcommand{\eea}{\end{eqnarray}}
\newcommand{\beas}{\begin{eqnarray*}}
	\newcommand{\eeas}{\end{eqnarray*}}

\title{Concentration of quantum states from \\quantum functional and transportation cost inequalities.}
\author[1,2]{Nilanjana Datta}
\author[1]{Cambyse Rouz\'{e}}
\affil[1]{\small Statistical Laboratory, Centre for Mathematical Sciences, University of Cambridge, Cambridge~CB30WB, UK}
\affil[2]{\small DAMTP, Centre for Mathematical Sciences, University of Cambridge, Cambridge~CB30WA, UK}
\begin{document}

 	\bibliographystyle{abbrv}

 	\maketitle	
\begin{abstract}
Quantum functional inequalities (e.g.~the logarithmic Sobolev- and Poincar\'{e} inequalities) have found widespread application in the study of the behavior of primitive quantum Markov semigroups. The classical counterparts of these inequalities are related to each other via a so-called transportation cost inequality of order 2 (TC$_2$). The latter inequality relies on the notion of a metric on the set of probability distributions called the Wasserstein distance of order 2. (TC$_2$) in turn implies a transportation cost inequality of order 1 (TC$_1$). In this paper, we introduce quantum generalizations of the inequalities (TC$_1$) and (TC$_2$), making use of appropriate quantum versions of the Wasserstein distances, one recently defined by Carlen and Maas and the other defined by us. We establish that these inequalities are related to each other, and to the quantum modified logarithmic Sobolev- and Poincar\'{e} inequalities, as in the classical case. We also show that these inequalities imply certain concentration-type results for the invariant state of the underlying semigroup. We consider the example of the depolarizing semigroup to derive concentration inequalities for any finite dimensional full-rank quantum state. These inequalities are then applied to derive upper bounds on the error probabilities occurring in the setting of finite blocklength quantum parameter estimation.
\end{abstract}
\section{Introduction}
\textit{Functional} and \textit{transportation cost} (also known as \textit{Talagrand}) \textit{inequalities} constitute a powerful set of mathematical tools which have found applications in various fields 
of mathematics and theoretical computer science. They are of relevance
in the analysis of the mixing time of a classical, continuous time, primitive Markov chain, which is the time it takes for the chain to come close in trace distance to its invariant distribution, starting from an arbitrary initial state. They are also used to derive {\em{concentration (of measure) inequalities}}, which provide upper bounds on the probability that a random variable, distributed according to a given measure, deviates from its mean (or median) by a given amount (say $r$). The concentration is
said to be {\em{exponential}} (resp.~{\em{Gaussian}}) if this probability decreases exponentially in $r$ (resp.~$r^2$). Concentration inequalities are of fundamental 
importance in fields as diverse as probability, statistics, convex geometry, functional analysis, statistical physics, and information theory. \\\\
There are different methods for deriving concentration inequalities. The \textit{entropy method}, introduced by Ledoux \cite{L97}, is an information-theoretic method and employs a particular class of functional inequalities, called \textit{logarithmic Sobolev inequalities} (or log-Sobolev inequalities, in short). Concentration inequalities can be shown to be fundamentally geometric in nature. An alternative method to derive them, which is due to Marton~\cite{[M96]}, exploits this fact by working at the level of probability measures on a metric probability space instead of functions, thus bypassing functional inequalities entirely. In this setting, one can introduce so-called \textit{Wasserstein distances} on the set of probability measures, which are relevant for the derivation of concentration inequalities in this geometric framework. A transportation cost (or Talagrand) inequality is an inequality between a Wasserstein distance between two probability measures and an information-theoretic quantity relating them, namely, their relative entropy (or Kullback-Leibler divergence \cite{CT12})\footnote{This inequality should not be confused with the related notion of Talagrand concentration inequality, see e.g.~\cite{TL13}.}. Marton~\cite{[M96]} showed that a particular Talagrand inequality (denoted by TC$_1$ in \Cref{fig1}) implies Gaussian concentration. Talagrand inequalities have been used in various areas of mathematics including probability theory, functional analysis, partial differential equations and differential geometry (see \cite{[V08]} and references therein for an exhaustive survey).
\\\\
Let us briefly introduce some of the relevant classical functional- and transportation cost inequalities in the discrete space setting (see \cite{[EM12]}): Consider a primitive, continuous time Markov chain on a finite metric space $(\Omega,d)$, with unique invariant distribution $\tilde{q}$, transition matrix $P$, generator $L=(P-\mathbb{I})$ and associated Markov semigroup $(P_t)_{t \geq 0 }:= (e^{tL})_{t\geq 0}$. The semigroup is said to satisfy a {\em{modified logarithmic Sobolev inequality}} (denoted as MLSI in \Cref{fig1}), with positive constant $\alpha_1$, if for any other distribution $q$,
\begin{align}\tag{MLSI}\label{logsob}
	2\alpha_1 D(q\|\tilde{q})\le I_{\tilde{q}}(q),
\end{align}
where $D(q\|\tilde{q}):=\sum_{\omega\in\Omega} q(\omega)\log(q(\omega)/\tilde{q}(\omega))$ is the relative entropy between $q$ and $\tilde{q}$,
and the quantity $I_{\tilde{q}}(q):=-\left.\frac{d}{dt}\right|_{t=0} D(P_t(q)\|\tilde{q})$ is the so-called (discrete) Fisher information of $q$ \cite{[EM12]}.
Diaconis and Saloff-Coste showed in \cite{DS96} that such an inequality leads to the so-called rapid mixing property of the semigroup:
\begin{align}\label{expconv}
	||q_t - {\tilde{q}}||_1 \leq \sqrt{ 2D(q_t\|{\tilde{q}})}\le \e^{-\alpha_1 t}\sqrt{D(q\|{\tilde{q}})},
\end{align}
where $q_t$ denotes the evolved probability distribution at time $t$, if the initial distribution of the Markov chain is $q$.\\\\
Similarly, for a given number $p\ge 2$, for the above Markov semigroup, a \textit{transportation cost inequality of order $p$} with positive constant $c_p$ is satisfied if for any distribution $q$, 
\begin{align}\tag{Tp}\label{tpp}
	W_p(q,\tilde{q})\le \sqrt{2c_pD(q\|\tilde{q})},
\end{align}
where $W_p(q,\tilde{q})$ is the so-called Wasserstein distance of order $p$ (or $p$-Wasserstein distance) on the set of probability distributions. When a transportation cost inequality of order $1$ (denoted as TC$_1$ in \Cref{fig1}) is satisfied, the invariant distribution~$\tilde{q}$ satisfies the following \textit{Gaussian concentration} property: for a random variable $X$ with law ${\tilde{q}}$, and any $k$-Lipschitz function $f:\Omega\mapsto\RR $,\footnote{That is, for all $x,y\in\Omega$, $|f(x)-f(y)|\le k\,d(x,y)$.}
\begin{align}\tag{Gauss}
	\mathbb{P}_{\tilde{q}}(f(X)-\EE[f(X)]\ge r)\le \exp\left(-\frac{r^2}{2c_1 {k}^2}\right).
\end{align}	
Moreover, a transportation cost inequality of order $2$ (denoted as TC$_2$ in \Cref{fig1}) implies another kind of functional inequality, namely the \textit{Poincar\'{e} inequality}:
\begin{align}\tag{P}
\lambda \operatorname{Var}_{\tilde{q}}(f(X))\le  -\sum_{\omega\in\Omega }  f(\omega)L(f)(\omega) \tilde{q}(\omega),
\end{align}
 where $\lambda>0$ is the \textit{spectral gap} of the generator of the underlying Markov chain, that is the absolute value of the second largest eigenvalue of its generator $L$. The Poincar\'{e} inequality in turn implies \textit{exponential concentration}:
\begin{align}\tag{Exp}
	\mathbb{P}_{\tilde{q}}(f(X)-\EE[f(X)]\ge r)\le 3\exp\left(-r\sqrt{\lambda}/(2k)\right).
\end{align}	
It also yields a weaker form of convergence in the analysis of the mixing time of the Markov chain, 
than the one provided in \reff{expconv}. The relations between the above inequalities is shown in \Cref{fig1}. 
\begin{figure}[!htbp]
	\[\begin{tikzcd}
	\text{MLSI}\arrow[rr, Rightarrow, "\cite{[OV00],[EM12]}"]&& \text{TC$_2$} \arrow[d,Rightarrow, "\cite{[EM12]}"] \arrow[rr, Rightarrow, "\cite{[OV00],[EM12]}"]&&\text{PI}\arrow[rr, Rightarrow, "\cite{[GM83]}"]&&\text{Exp}\\
	&& \text{TC$_1$}\arrow[rr, Rightarrow, "\cite{[BG99],[EM12]}"]&&\text{Gauss}&
	\end{tikzcd}
	\]
	\caption{Chain of classical functional- and transportation cost inequalities and related concentrations. The citations above the arrows refer to the papers in which the implications were proved.}
	\label{fig1}
\end{figure}
\\\\Quantum analogues of modified log-Sobolev and Poincar\'{e} inequalities, as well as their applications to quantum information theory, have recently attracted a lot of attention (see e.g. \cite{[OZ99],[CS08],TKRWV10,[M12],[KT13],[TPK14],[CKMT15],[KT16],[MF16],[DB14],[BK16],[CM16],[CM14],[JZ15]}). In particular, Kastoryano and Temme \cite{[KT13]} defined a non-commutative version of the modified log-Sobolev inequality and showed that it implies a Poincar\'e inequality. More recently, Carlen and Maas \cite{[CM14],[CM16]} defined a {\em{quantum Wasserstein distance of order $2$}}, which turns out to provide the manifold of full-rank states with a Riemannian metric. In \cite{[CM16]}, they proved that, for a class of quantum Markov semigroups with unique invariant state $\sigma$, the modified log-Sobolev inequality holds provided the entropic functional $D(.\|\sigma)$ is convex along geodesics in the manifold of full-rank states equipped with this metric. A non-commutative analogue of the  Wasserstein distance of order $1$ has also been introduced and studied by Junge and Zeng \cite{[JZ15]} in the context of von Neumann algebras equipped with a tracial state.
\paragraph{Our contribution:}
In this paper, we define quantum analogues of the transportation cost inequalities or orders $1$ and $2$ and prove, similarly to the classical setting, the chain of implications given in \Cref{fig2}.
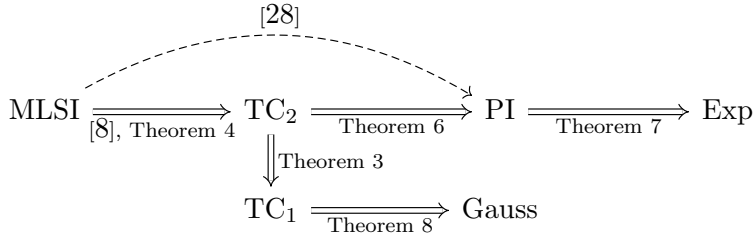
\begin{figure}[!ht]
	\[
	\begin{tikzcd}
	\text{MLSI}\arrow[rrrr, bend left, dashrightarrow, "\cite{[KT13]}"]\arrow[rr, swap, Rightarrow, "\cite{[CM14]} \text{, \Cref{logsobtalagrand}}"]&& \text{TC$_2$} \arrow[d,Rightarrow, "\text{\Cref{t2t1}}"] \arrow[rr, Rightarrow, swap, "\text{\Cref{talagrandpoincare}}"]&&\text{PI}\arrow[rr, swap, Rightarrow, "\text{\Cref{poincare}}"]&&\text{Exp}\\
	&& \text{TC$_1$}\arrow[rr, Rightarrow,swap, "\text{\Cref{talconc}}"]&&\text{Gauss}&
	\end{tikzcd}
	\]
	\caption{Chain of quantum functional- and transportation cost inequalities and related concentrations. The implication MLSI $\implies$ PI was proved by Kastoryano and Temme in~\cite{[KT13]}.}
	\label{fig2}
\end{figure}
   \noindent 
 Note that the implication MLSI$\Rightarrow$TC$_2$ was proved in \cite{[CM14]} in the case of the fermionic Fokker-Planck semigroup on the Clifford algebra. {In order to define TC$_1$, we introduce a new notion of quantum Wasserstein distance of order $1$. For sake of simplicity, we restrict our attention to the case of quantum Markov semigroups defined on the space $\mathcal{B}(\mathcal{H})$ of linear operators defined on a finite-dimensional Hilbert space $\mathcal{H}$. However, we expect our results to hold for more general cases, and in particular for quantum evolutions on separable Hilbert spaces.
\paragraph{Layout of the paper:}
In \Cref{sec2}, we introduce the necessary notations and definitions, including relative modular operators, quantum Markov semigroups and quantum Wasserstein distances. Quantum transportation cost inequalities are introduced in \Cref{qfunc} and the relations between them are proved. More precisely, we prove TC$_2$ $\Rightarrow$ TC$_1$ in 
\Cref{t2t1}, MLSI $\Rightarrow$ TC$_2$ in \Cref{logsobtalagrand}, and TC$_2$ $\Rightarrow$ PI in \Cref{talagrandpoincare}. In \Cref{qconc}, we derive two types of quantum concentration inequalities from the transportation cost and Poincar\'{e} inequalities, namely Gaussian (\Cref{talconc}) and exponential (\Cref{poincare}) concentration, respectively. Finally, we apply our concentration inequalities to the problem of estimating the parameter of a quantum state in the finite blocklength setting in \Cref{parameter}.
\section{Notations and preliminaries}\label{sec2}
\subsection{Operators, states and modular theory}
Let $\cH$ denote a finite-dimensional Hilbert space, let $\cB(\cH)$ denote the algebra of linear operators acting on $\cH$ and $\cB_{sa}(\cH) \subset \cB(\cH)$ the subspace of self-adjoint operators. Let $\cP(\cH)$ be the cone of positive semi-definite operators on $\cH$ and $\cP_{+}(\cH) \subset  \cP(\cH)$ the set of (strictly) positive operators. Further, let $\cD(\cH):=\lbrace\rho\in\cP(\cH)\mid \tr\rho=1\rbrace$ denote the set of density operators (or states) on $\cH$, and $\cD_+(\cH):=\cD(\cH)\cap \cP_+(\cH)$ denote the subset of full-rank states. We denote the support of an operator $A$ as ${\mathrm{supp}}(A)$ and the range of a projection operator $P$ as ${\mathrm{ran}}(P)$. Let $\mathbb{I}\in\cP(\cH)$ denote the identity operator on $\cH$, and $\id:\cB(\cH)\mapsto \cB(\cH)$ the identity map on operators on~$\cH$. A linear map $\Lambda:\cB(\cH)\to \cB(\cH)$ is said to be unital if $\Lambda(\mathbb{I})=\mathbb{I}$. The operator norm on $\cB(\cH)$ is defined as $\|A\|_\infty:=\sup_{\psi\in\cH\backslash\{0\}} \|A\psi\|/\|\psi\|$, and the operator norm for a superoperator $\Lambda:\cB(\cH)\to\cB(\cH)$ is defined as $\|\Lambda\|_{\infty\to\infty}:=\sup_{A\in\cB(\cH)}\| \Lambda(A)\|_\infty/\|A\|_\infty$ . The Hilbert Schmidt inner product between two operators $A,B\in\cB(\cH)$ is defined as $\langle A,B\rangle_{_{HS}} :=\tr(A^* B)$. For any $p\ge 1$, the $p$-Schatten norm is defined as $\|A\|_p:=(\tr|A|^p)^{1/p}$. Any observable $f\in\cB_{sa}(\cH)$ has a \textit{spectral decomposition} of the form $f = \sum_{\lambda \in \spec(f)} \lambda \,P_\lambda(f),$ where $\spec(f)$ denotes the spectrum of $A$, and $P_\lambda(f)$ is the projection operator corresponding to the eigenvalue $\lambda$. For any interval $E$ of $\RR$, we denote by $\mathbf{1}_E(f)$ the projection $\sum_{\lambda\in E\cap \spec(f)}P_\lambda(f)$ of any self-adjoint operator $f$. Unless stated otherwise, we adopt the following convention: operators are denoted by Roman letters, and functions by Greek letters. Given two full-rank states $\rho,\sigma$, the \textit{relative modular operator} $\Delta_{\rho|\sigma}$ is defined as the map
\begin{equation}\label{relmod}
	\begin{array}{cccc}
		\Delta_{\rho|\sigma}:
		 & \mathcal{B}(\cH) & \to & \hspace{-1em} \cB(\cH) \\
		& f & \mapsto & \rho f \sigma^{-1}
	\end{array}
\end{equation}
As a linear operator on the Hilbert space $(\mathcal B(\cH),\langle .,.\rangle_{_{HS}})$ obtained by equipping $\cB(\cH)$ with the Hilbert-Schmidt inner product, $\Delta_{\rho|\sigma}$ is positive and its spectrum $\spec(\Delta_{\rho|\sigma})$ consists of the ratios of eigenvalues $\lambda/\mu$, $ \lambda \in \spec(\rho)$, 
$ \mu \in \spec(\sigma)$. For any $x\in \spec (\Delta_{\rho|\sigma})$, the corresponding spectral projection is the map
\begin{equation} \label{eq_specprojDelta}
	\begin{array}{cccc}
		P_{x}(\Delta_{\rho|\sigma}): & \cB(\cH) & \to & \hspace{-1em} \cB(\cH) \\
		& f& \mapsto & \underset{\lambda\in \spec(\rho),\mu\in\spec(\sigma):\lambda/\mu=x}{\sum}P_\lambda(\rho) f  P_\mu(\sigma).
	\end{array}
\end{equation}
It can be readily verified that $\Delta_{\rho|\sigma}^s(f)=\rho^s f \sigma^{-s}$ for any $s$ in $[0,1]$. In the particular case $\rho=\sigma$, $\Delta_{\sigma}:=\Delta_{\sigma|\sigma}$ is the so-called \textit{modular operator}, defined as
\begin{align*}
	\Delta_\sigma(f):=\Delta_{\sigma|\sigma} (f)\equiv \sigma f \sigma^{-1},~~ f\in\cB(\cH).
\end{align*}
The \textit{modular automorphism group} $(\alpha_t)_{t\in\RR}$ of a state $\sigma\in\cD_+(\cH)$ is defined as follows
\begin{align} \label{mod}
	\alpha_t(f):=\e^{it H}f\e^{-it H},
\end{align}
where $H:=-\log\sigma$. Note that $\Delta_\sigma:=\alpha_i$.
\subsection{Inner products}
An inner product $\langle .,.\rangle$ on $\cB(\cH)$ is said to be \textit{compatible} with a state $\sigma\in\cD_+(\cH)$ if for all $f\in \cB(\cH)$, $\tr(\sigma f)=\langle \mathbb{I},f\rangle$. This notion is useful as it allows to show, analogously to the classical case, that for any completely positive, unital map $\Phi$ that is self-adjoint with respect to an inner product $\langle .,. \rangle$\footnote{That is for which $\langle f,\Phi(g)\rangle=\langle \Phi(f),g\rangle$ for all $f,g\in\cB(\cH)$} compatible with $\sigma\in \cD_+(\cH)$, $\sigma$ is invariant with respect to $\Phi$. Indeed, using the fact that $\Phi$ is unital, we get:
\begin{align*}
	\tr(\sigma f)=\langle \mathbb{I},f\rangle =\langle \Phi (\mathbb{I}),f\rangle=\langle  \mathbb{I}, \Phi(f)\rangle =\tr(\sigma\, \Phi(f)), ~~\forall f\in \cB(\cH).
\end{align*}
 Given any function $\varphi:(0,\infty)\to (0,\infty)$ and a state $\sigma\in \cD_+(\cH)$, one can easily check that the following quadratic form defines an inner product:
\begin{align}\label{innergen}
	\langle f,g\rangle_{\varphi,\sigma}:= \tr(f^* R_\sigma\circ \varphi(\Delta_\sigma)(g) ),
\end{align}
where $R_\sigma:\cB(\cH)\to \cB(\cH)$ is the operator of right multiplication by $\sigma$.
Perhaps the two most commonly used functions $\varphi$ are $\varphi_1:t\to 1$ and $\varphi_{1/2}: t\to t^{1/2}$, for which the inner product defined in \Cref{innergen} reduces to:
\begin{align*}
	&\langle f,g\rangle_{\varphi_1,\sigma}=   \tr(\sigma f^* g)\equiv	\langle f,g\rangle_{1,\sigma},\\
	&\langle f,g\rangle_{\varphi_{1/2},\sigma}=   \tr(\sigma^{1/2} f^* \sigma^{1/2} g)\equiv	\langle f,g\rangle_{1/2,\sigma}.
\end{align*}
More generally, for $s\in [0,1]$, and $\varphi_s:= t\to t^{1-s}$, one defines the following inner product that is compatible with $\sigma$:
\begin{align*}
	\langle f,g\rangle_{\varphi_s,\sigma}=   \tr(\sigma^s f^*\sigma^{1-s} g)\equiv	\langle f,g\rangle_{s,\sigma}.
\end{align*}
The case $s=1/2$ is special because it is the only one for which for any completely positive (CP) map $\Phi:\cB(\cH)\to \cB(\cH)$, its adjoint $\Phi^{(*,1/2)}$ with respect to $\langle .,.\rangle_{1/2,\sigma} $, defined through the following identity
\begin{align*}
	\langle f,\Phi (g)\rangle_{1/2,\sigma}=\langle \Phi^{(*,1/2)}(f),g\rangle_{1/2,\sigma},
\end{align*}
is also CP.
It is not in general true that self-adjointness is a notion independent of the inner product chosen. However, the following theorem, for which a proof can be found in \cite{[CM16]}, gives a sufficient condition for this to be true. 
\begin{theorem}[see \cite{[CM16]} Theorem 2.9]\label{selfadjointness}
	Let $\sigma$ be a non-degenerate full-rank density matrix, and let $\LL$ be any linear map on $\cB(\cH)$. Then:
	\begin{itemize}
		\item[1] If $\LL$ is self-adjoint with respect to the inner product $\langle .,.\rangle_{1,\sigma}$ and adjoint-preserving (i.e. $\LL(f^*)=(\LL(f))^*$ for all $f\in\cB(\cH)$), then $\LL$ commutes with the modular automorphism group of $\sigma$:
		\begin{align*}
			\alpha_t\circ \LL=\LL\circ \alpha_t,~~\forall t\in\CC,
		\end{align*}
		where $\alpha$ is defined in \Cref{mod}, and $\LL$ is self-adjoint with respect to $\langle.,.\rangle_{\varphi,\sigma}$, for all functions $\varphi:(0,\infty)\to (0,\infty)$.
		\item[2] If $\LL$ commutes with the modular automorphism group of $\sigma$, and $\LL$ is self-adjoint with respect to $\langle.,.\rangle_{\varphi,\sigma}$ for some $\varphi:(0,\infty)\to (0,\infty)$, then $\LL$ is self-adjoint with respect to $\langle .,.\rangle_{\psi,\sigma}$ for all $\psi:(0,\infty)\to (0,\infty)$.
	\end{itemize}
\end{theorem}
\begin{remark} Obviously, a linear map $\LL$ on $\cB(\cH)$ commutes with the modular automorphism group of a state $\sigma$ if and only if $\LL$ commutes with $\Delta_\sigma$.
\end{remark}
\begin{remark}
	Given a full-rank state $\sigma$, a linear map $\LL$ is self-adjoint with respect to $\langle .,.\rangle_{1/2,\sigma}$ if and only if its adjoint $\LL_*$ with respect to $\langle .,.\rangle_{_{HS}}$ can be expressed as:
	\begin{align}\label{ll*}
		\LL_*=\Gamma_\sigma\circ \LL\circ \Gamma_\sigma^{-1}.
	\end{align}
	where $\Gamma_\sigma(f):=\sigma^{1/2}f\sigma^{1/2}$. Indeed, for any $f,g\in\cB(\cH)$:
	\begin{align*}
		\langle f&,\LL(g)\rangle_{1/2,\sigma}=\langle \LL(f),g\rangle_{1/2,\sigma}\\
		&	\Leftrightarrow \tr(\sigma^{1/2}f^*\sigma^{1/2}\LL(g))=\tr(\sigma^{1/2}\LL(f)^*\sigma^{1/2}g)=\tr(f^*\LL_*(\sigma^{1/2}g\sigma^{1/2})).
	\end{align*}
\end{remark}
\subsection{Entropic quantities}
A quantum generalization of the Kullback-Leibler divergence is given by Umegaki's quantum relative entropy which is defined as follows: for two states $\rho,\sigma\in\cD(\cH)$:
\begin{align*}
	D(\rho\|\sigma):= \left\{\begin{aligned}
		&	\tr(\rho(\log\rho-\log\sigma))~~~~~\text{if} \supp(\rho)\subseteq \supp(\sigma),\\
		&	0~~~~~~~~~~~~~~~~~~~~~~~~~~~~~\text{else.}
	\end{aligned}
	\right.
\end{align*}
It is related to another entropic quantity, namely the $L^1$-relative entropy defined in \cite{[KT13]}, as follows:
\begin{align*}
	D(\rho\|\sigma)= \operatorname{Ent}_{1,\sigma}(\Gamma_\sigma^{-1}(\rho)),
\end{align*}
where
\begin{align*}
	\operatorname{Ent_{1,\sigma}}(f):=\tr(\Gamma_\sigma(f)(\log(\Gamma_\sigma(f))-\log\sigma))-\tr(\Gamma_\sigma(f))\log(\tr(\Gamma_\sigma(f))).
\end{align*}
There are various other quantum generalizations of the Kullback Leibler divergence (see e.g.~ \cite{[HM16]} and references therein). In \cite{[PR97]}, Petz and Ruskai introduced the following generalization, which they called the quantum \textit{maximal divergence}:
\begin{align}\label{maxi}
	\widehat{D}(\rho\|\sigma):=\tr[\sigma\,\Gamma_\sigma^{-1}(\rho)~\log(\Gamma_\sigma^{-1}(\rho))]\equiv \tr[\sigma\,(\sigma^{-1/2}\rho~\sigma^{-1/2}\log(\sigma^{-1/2}\rho~\sigma^{-1/2}))] .
\end{align}
This quantum divergence is maximal among the so-called quantum $f$-divergences (see \cite{[M16]}). In particular, for any $\rho,\sigma\in \cD_+(\cH)$,
\begin{align}\label{DhD}
	D(\rho\|\sigma)\le \widehat{D}(\rho\|\sigma).
\end{align} 
\subsection{Quantum Markov semigroups}\label{sec-QDS}
\subsubsection*{The GKLS form}
Given a Hilbert space $\cH$, a \textit{quantum Markov semigroup} models the evolution of a memoryless quantum system with Hilbert space $\cH$. In the Heisenberg picture, it is given by a one-parameter family $\left(\Lambda_t\right)_{t \geq 0}$ of linear, completely positive, unital maps on $\cB(\mathcal H)$ satisfying the following properties
\begin{itemize}
	\item $\Lambda_0 = {\rm{id}}$;
	\item $\Lambda_t \Lambda_s = \Lambda_{t+s}$ $\,-\,$ semigroup property;
	\item $\underset{t\to 0}{\lim} || \Lambda_t(f) - f||_\infty = 0$ $\,-\,$ strong continuity,
\end{itemize}
The parameter $t$ plays the role of time. Hence, $\Lambda_t$ results in time evolution over the interval $[0,t]$, and is called a {\em{quantum dynamical map}}. For each quantum Markov semigroup there exists an operator $\LL$ called the \textit{generator}, or \textit{Lindbladian}, of the semigroup, such that
\begin{align}\label{llambda}
	\frac{d}{dt}\Lambda_t=\Lambda_t\circ \LL=\LL\circ\Lambda_t.
	\end{align}
Reciprocally, the Hille-Yosida theorem \cite{[P83]} states that one can uniquely associate to a given operator $\LL$ a semigroup $(\Lambda_t)_{t\ge 0}$ such that \Cref{llambda} is satisfied under some regularity conditions on the resolvent of $\LL$. Formally, one writes $\Lambda_t = e^{t\cL}$. The semigroup property embodies the assumptions of time-homogeneity and Markovianity, since it implies that the time evolution is independent of its history and of the actual time. Quantum Markov semigroups provide a generalization of classical Markov chains (in the discrete state space setting) or Markov semigroups associated to a Markov process (in the continuous state space setting), and therefore it is interesting to study questions similar to the ones asked in the classical framework, namely existence and uniqueness of invariant states, detailed balance conditions, convergence towards the invariant state, speed of convergence etc.
	\\\\
	A state $\rho$ is said to be an {\em{invariant state}} with respect to the semigroup $(\Lambda_t)_{t\ge 0}$ if $\Lambda_{*t}(\rho)=\rho$ for all time $t\ge 0$, where $\Lambda_{*t}$ is the adjoint of the map $\Lambda_t$, for each $t$, or equivalently
\begin{align*}
	\tr(\rho \Lambda_t(f))=	\tr(\rho f),\qquad \forall \, f\in\mathcal{B}(\cH), \,\, \forall t \geq 0.
\end{align*}
Semigroups admitting a unique, full-rank invariant state are called \textit{primitive}.
A quantum Markov semigroup is said to admit a \textit{GKLS} form (after Gorini, Kossakowski, Lindblad and Sudarshan) if its generator can be written as follows:
\begin{align}\label{GKLS}
	\LL(f)=i[H,f] -\frac{1}{2}\sum_{j\in \mathcal{J}} (L_j^*L_j f-2 L_j^* f L_j+fL_j^*L_j),
	\end{align}
	where $\mathcal{J}$ is a finite index set, and $H, L_j,~j\in\mathcal{J}$ are operators on $\cH$, with $H$ being self-adjoint. The operator $H$ is the Hamiltonian governing the unitary evolution of the system, whereas the so-called \textit{Lindblad} operators, $L_j$, govern the dissipative part of the evolution, resulting from the interaction of the system with its environment. In a series of seminal articles (see \cite{[GKS76],[L76a],[C79]}), Gorini, Kossakowski, Lindblad and Sudarshan proved that any uniformly continuous quantum Markov semigroup acting on $\cB(\cH)$ can be written in this form, and for such semigroups the operators $H, L_j,~j\in\mathcal{J}$ are bounded. In the context of quantum optics, many more generators associated to semigroups which are not uniformly continuous can also be written in the GKLS form (see e.g. \cite{[CS08],[L76b]}). \Cref{GKLS} is a quantum generalization of the classical case of continuous time Markov chains on finite state spaces.
\subsubsection*{Detailed balance}
In the classical setting, the \textit{detailed balance condition} of a Markov chain characterizes the time reversibility of the process. It means that, starting from a given invariant measure, the probability to get to the state $i$ of the chain starting from the state $j$ is equal to the probability to go from $i$ to $j$. This condition is equivalent to the symmetry of the generator of the Markov chain with respect to the weighted inner product associated to the given invariant state: more precisely let $(P_t)_{t\ge 0}$ be a classical continuous time Markov semigroup on a finite state space $\Omega\cong \{1,\dots,N\}$, with associated generator $L$. Then $(P_t)_{t\ge 0}$ satisfies detailed balance with respect to a probability measure $\pi$ on $\Omega$ if:
\begin{align*} 
\langle f,L( g)\rangle_{\pi}=\langle L( f), g\rangle_{\pi},
\end{align*}
for any pair of functions $f,g:\Omega\to \RR$, where the inner product above is defined as follows:
\begin{align*} 
	\langle f,g\rangle_{\pi}=\sum_{x\in \Omega} \pi(x)f(x)g(x).
\end{align*}
This directly implies that the measure $\pi$ is invariant with respect to the semigroup $(P_t)_{t\ge 0}$. Unlike the classical case, there are infinitely many ways to define a detailed balance condition in the quantum setting, all of which reduce to the classical one in the case of commutative algebras. In what follows, we are only concerned with quantum Markov semigroups $(\Lambda_t)_{t\ge 0}$ whose generators are self-adjoint with respect to $\langle .,.\rangle_{1,\sigma}$, that is
\begin{align}
	\langle g,\LL (f)\rangle_{1,\sigma}=\langle \LL(g),f\rangle_{1,\sigma}.
	\end{align}
This condition is also referred to as the \textit{quantum detailed balance condition}. The following theorem states that under this assumption, the generator $\LL$ of $(\Lambda_t)_{t\ge 0}$ admits a special GKLS form (first proved in a more general context by Alicki under the assumption that $\sigma$ has non-degenerate spectrum in \cite{A76}, see also \cite{KFGV77}):
\begin{theorem}[\cite{[CM16]} Theorem 3.1]\label{theo2}Let $(\Lambda_t)_{t\ge0}$ be a quantum Markov semigroup on $\cB(\cH)$, where $\cH$ is a $d$ dimensional Hilbert space, and let $\sigma$ a full-rank state. Suppose that the generator $\LL $ of $(\Lambda_t)_{t\ge0}$ is self-adjoint with respect to $\langle .,.\rangle_{1,\sigma}$. Then the generator $\LL$ of $(\Lambda_t)_{t\ge 0}$ has the following form: $\forall f\in \cB(\cH)$, 
					\begin{align}\label{LLDBC}
						\LL(f)&=\sum_{j\in \mathcal{J}}c_j\left( \e^{-\omega_j/2}\tilde{L}_j^*[f,\tilde{L}_j]+\e^{\omega_j/2}[\tilde{L}_j,f]\tilde{L}_j^*\right)\\
						&=\sum_{j\in \mathcal{J}}c_j\e^{-\omega_j/2}\left( \tilde{L}_j^*[f,\tilde{L}_j]+[\tilde{L}_j^*,f]\tilde{L}_j\right),
					\end{align}
					where $\cJ$ is a finite set of cardinality $|\cJ|\le d^2-1$, $\omega_j\in\RR$ and $c_j>0$ for all $j\in\mathcal{J}$, and $\{\tilde{L}_j\}_{j\in\mathcal{J}}$ is a set of operators in $\cB(\cH)$ with the properties:
					\begin{itemize}
						\item[1] $\{\tilde{L}_j\}_{j\in\mathcal{J}}=\{\tilde{L}_j^*\}_{j\in\mathcal{J}}$;
						\item[2] $\{\tilde{L}_j\}_{j\in\mathcal{J}}$ consists of eigenvectors of the modular operator $\Delta_\sigma$ with
						\begin{align}\label{eigenD}
							\Delta_\sigma(\tilde{L}_j)=\e^{-\omega_j}\tilde{L}_j;
						\end{align}
						\item[3] $\frac{1}{\dim(\cH)}\tr(\tilde{L}_j^*\tilde{L}_k)=\delta_{k,j}$ for all $j,k\in\mathcal{J}$;
						\item[4] $\tr(\tilde{L}_j)=0$ for all $j\in\mathcal{J}$.
						\end{itemize}
					Finally for any $j,j'\in\mathcal{J}$
					\begin{align}\label{symcond}
						c_j=c_{j'}~~\text{ when }\tilde{L}_j^*=\tilde{L}_{j'}.
					\end{align}
					Conversely,  given any full-rank state $\sigma$, any set $\{\tilde{L}_j\}_{j\in\mathcal{J}}$ satisfying the above four items for some $\{\omega_j\}_{j\in \mathcal{J}}\subset \RR$ and any set $\{c_j\}_{j\in\mathcal{J}}$ of positive numbers satisfying the symmetry condition \eqref{symcond}, the operator $\LL$ given by \Cref{LLDBC} is the generator of a quantum Markov semigroup $(\Lambda_t)_{t\ge 0}$ which is self-adjoint with respect to $\langle .,.\rangle_{1,\sigma}$. 
				\end{theorem}			
				\begin{remark}
				The operators ${L}_j\equiv \tilde{L}_j\sqrt{2c_j}\e^{-\omega_j/4}$ in \Cref{LLDBC} play the role of the Lindblad operators, and the Hamiltonian of the system is equal to $0$.
				\end{remark}
			\begin{remark}\label{rmkU}
				Given a primitive quantum Markov semigroup with generator which is self-adjoint with respect to $\langle .,.\rangle_{1,\sigma}$, the operators $\tilde{L}_j$ and constants $c_j$ in \Cref{LLDBC} are uniquely defined up to a unitary transformation, i.e. for any two representations $(\{ c_j\},\{\tilde{L}_j\},\{\omega_j\})$ and $(\{ c_j'\},\{\tilde{L}'_j\},\{\omega_j'\})$ of $\LL$, there exists a $|\cJ|\times |\cJ|$ unitary matrix $U$ such that, unless $\omega_j=\omega_k'$, $U_{jk}=0$ and
				\begin{align}\label{eq17}
				&\sum_{k\in\cJ}\overline{U}_{kj}c_k' U_{kl}=c_l~\delta_{jl},\qquad\tilde{L}_l'=\sum_{k\in\cJ}U_{lk}\tilde{L}_k,
				\end{align}  
			where $\overline{a}$ denotes the complex conjugate of $a\in\CC$.
			\end{remark}	
				\subsubsection*{Dirichlet forms and the quantum de Bruijn identity}
	The $2-$\textit{Dirichlet form} of a quantum Markov semigroup is a very useful object which completely characterizes the semigroup. It is defined as follows: For a semigroup $(\Lambda_t)_{t\ge 0}$ with generator $\LL$, invariant state $\sigma\in\cD(\cH)$, and $\varphi:(0,\infty)\to (0,\infty)$,
	\begin{align}\label{2dirichlet}
		\cE_{\varphi,2}(f,g):= -\langle f,\LL(g)\rangle_{\varphi,\sigma}.
		\end{align} 
		In the context of quantum logarithmic Sobolev inequalities, one can similarly define $p$-Dirichlet forms for $p> 1$ and $\varphi=\varphi_{1/2}$ (see \cite{[OZ99]}). In particular, the $1$-Dirichlet form, first defined in \cite{[KT13]} as the limit $p\to 1$ of the family of $p$-Dirichlet forms, is particularly useful:
		\begin{align*}
			\cE_{1}(f,f):= -\frac{1}{2}\tr\left( \Gamma_\sigma(\LL(f))(\log(\Gamma_{\sigma}(f))-\log\sigma\right),
			\end{align*}
		 $\cE_1(f,f)$ can be related to the so-called \textit{entropy production} associated to the semigroup $(\Lambda_t)_{t\ge 0}$ \cite{S78}. This name is justified by the following identity.
		\begin{lemma}[quantum de Bruijn's identity]
Given a semigroup $(\Lambda_t)_{t\ge 0}$ with generator $\LL$, self-adjoint with respect to $\langle .,.\rangle_{1/2,\sigma}$ for a given full-rank state $\sigma$,
	\begin{align}\label{debruijn}
		\frac{d}{dt}D(\Lambda_{*t}(\rho)\|\sigma)=-2\cE_1(\Gamma_{\sigma}^{-1}\circ\Lambda_{*t}(\rho),\Gamma_\sigma^{-1}\circ\Lambda_{*t}(\rho)),~~~~~\forall t\ge 0.
	\end{align}	
\end{lemma}
	\begin{proof}
		Define $\rho_t:=\Lambda_{*t}(\rho)$. By Theorem 3 of \cite{S78},
		\begin{align*}
			\frac{dD(\rho_t\|\sigma)}{dt}=\tr(\LL_*(\rho_t)(\log \rho_t-\log \sigma)).
			\end{align*}
	The result follows from \reff{ll*}.
	\qed
	\end{proof}	
The quantity on the right hand side of \Cref{debruijn} is called \textit{entropy production}:
	\begin{align}\label{Fisher}
		\operatorname{EP}_\sigma(\rho):=2\mathcal{E}_1(\Gamma_{\sigma}^{-1}(\rho),\Gamma_{\sigma}^{-1}(\rho))=: \operatorname{I}_\sigma(\rho).
	\end{align}	
 In analogy with the classical case, we also denote it by $\operatorname{I}_\sigma(\rho)$, and refer to it as the \textit{quantum Fisher information} of $\rho$ with respect to the sate $\sigma$.\\\\	
		 In the case of a semigroup $(\Lambda_t)_{t\ge 0}$ whose generator is self-adjoint with respect to $\langle .,.\rangle_{1,\sigma}$, for a given full-rank state $\sigma$ and $s\in [0,1]$, the Dirichlet form can be rewritten as follows \cite{[CM16]}: 
\begin{align}\label{2dirichletform}
	\mathcal{E}_{s,2}(f,g)\equiv	\mathcal{E}_{\varphi_s,2}(f,g)=\sum_{j\in\mathcal{J}}c_j \e^{(1/2-s)\omega_j}\langle \partial_j f,\partial_j g\rangle_{s,\sigma},
\end{align}
where 
\begin{align}\label{derivation}
	\partial_j f:=[\tilde{L}_j,f].
	\end{align}
The passage from \Cref{2dirichlet} to \Cref{2dirichletform} can be understood as some sort of non-commutative integration by parts, where the generator $\LL$ plays the role of the Laplacian, and each $\partial_i$ the role of the partial derivative ``along the direction $i$''. By analogy with the classical case of Markov processes defined on finite or continuous state spaces, the operators $\partial_i$ defined above are called \textit{derivations} \cite{Davies1993}. Classically, such an integration by parts enables one to map the problem of finding the rates of convergence of a primitive Markov  process with invariant measure $\mu$ to the study of the concentration of measure of bounded or Lipschitz functions of a given random variable $X$ of law $\mu$. The derivations defined in \Cref{derivation}, which directly originate from the generator of a semigroup whose expression is given in \Cref{LLDBC}, allow for the construction of a non-commutative differential calculus, initiated in \cite{[CM14],[CM16]}, which is analogous to the classical case, and is described as follows: Given an operator $U\in \cB(\cH)$, its non-commutative \textit{gradient} is defined as:
\begin{align*}
	\nabla U:=(\partial_1U,...,\partial_{|\mathcal{J}|}U), ~~~ U\in\cB(\cH).
\end{align*}	
 Moreover, given a vector $\mathbf{A}\equiv (A_1,...,A_{|\cJ|})\in \bigoplus_{j\in\cJ}\cB(\cH)$, the \textit{divergence} of $\mathbf{A}$ is defined as
\begin{align}\label{div}
	\operatorname{div}(\mathbf{A}):=\sum_{j\in\mathcal{J}} c_j[A_j,\tilde{L}_j^*].
\end{align}
For $\vec{\omega}:=(\omega_1,...,\omega_{|\mathcal{J}|})$, define the linear operator $[\rho]_{\vec{\omega}}$ on $\bigoplus_{j\in \mathcal{J}}\cB(\cH)$ through
\begin{align*}
	[\rho]_{\vec{\omega}}\mathbf{A}:=([\rho]_{\omega_1}A_1,...,[\rho]_{\omega_{|\mathcal{J}|}A_{|\mathcal{J}|}}),~~~~~~~~\mathbf{A}\equiv(A_1,...,A_{|\mathcal{J}|}),
\end{align*}
where for any $\omega_j\in \RR$,
\begin{align}\label{fomega}
	[\rho]_{\omega_j}=R_\rho\circ f_{\omega_j}(\Delta_{\rho}),~~~~f_{\omega_j}(t):=\e^{\omega_j/2}\frac{t-\e^{-\omega_j}}{\log t+\omega_j},~~~~t\in\RR.
\end{align}
One can easily verify that, up to a normalizing factor depending on $\omega_j$, $[\rho]_{\omega_j}(A)$ reduces to the operation of multiplication of $A$ by $\rho$ in the case when $A$ and $\rho$ commute. The following result can be found in \cite{[CM16]}:
\begin{lemma}[see Lemma 5.8 of \cite{[CM16]}]\label{lemma}
	For any $\omega\in\RR$, $\rho\in \cD_+(\cH)$, and $A\in\cB(\cH)$,
	\begin{align*}
		&[\rho]_{\omega}^{-1}(A)=\int_0^\infty (t+\e^{-\omega/2}\rho)^{-1}A(t+\e^{\omega/2}\rho)^{-1}dt.\\
		&[\rho]_{\omega}(A)=\int_0^1 \e^{\omega(1/2-s)}\rho^s A\rho^{1-s}ds.
	\end{align*}
\end{lemma}
 \noindent We end this section by stating a very useful lemma:
\begin{lemma}[Chain rule identity, see \cite{[CM16]} Lemma 5.5]\label{chainrule}
	For all $L\in \cB(\cH)$, and all $f\in\cB(\cH)$,
	\begin{align*}
		[L,\e^f]=\int_0^1 \e^{sf}~[L,f]~\e^{f(1-s)}ds.
	\end{align*}	
\end{lemma}
\begin{proof}
Define the function $\varphi:[0,1]\ni s\mapsto \e^{(1-s)f}L\e^{sf}$. Therefore, 
\begin{align*}
\frac{d}{ds} \e^{(1-s)f}L\e^{sf}=\e^{(1-s)f}[L,f]\e^{sf}.
\end{align*}
The result follows by integrating the above equation from $0$ to $1$.
	\qed
\end{proof}
\subsection{Quantum Wasserstein distances}\label{wass}
\subsubsection*{The quantum Wasserstein distance of order $2$}
Classically, a useful measure of distance between two probability measures $\mu,\nu$ defined on a metric space $(M,d)$ is the so-called \textit{Wasserstein distance} of order $p\ge 1$. It is defined as follows:
\begin{align}\label{2.44}
W_p(\mu,\nu):=\inf_{\pi\in \Pi(\mu,\nu)}\left(\int_{M}d(x,y)^p d\pi(x,y)\right)^{1/p},
\end{align}
where the infimum is taken over the set $\Pi(\mu,\nu)$ of probability measures $\pi$ on the Cartesian product $M\times M $ with marginals $\mu$ and $\nu$. In contrast to the Kullback-Leibler divergence, Wasserstein distances satisfy the properties of a distance, and therefore endow the set $\mathcal{P}(M)$ of probability measures on $M$ with a metric. The cases $p=1$ and $p=2$ are of particular importance. For $p=2$, it can be proved (under some technical assumptions) that $\mathcal{P}(M)$ can be turned into a Riemmanian manifold with associated Riemannian distance given by $W_2$. Moreover, $W_2$ can be re-expressed as the minimizer of a variational problem expressed in terms of a Markov evolution on the manifold $\mathcal{P}(M)$ (see \cite{[V08],[EM12],[RS14]} for more details). In the simplest case when $M$ is $\RR^n$ and $\mu$, respectively $\nu$, are absolutely continuous with respect to the Lebesgue measure, with respective densities $\varphi$ and $\psi$, this is given in terms of the Benamiou-Brenier formula \cite{[BB00]}:
\begin{align}\label{Benambren}
	W_2(\mu,\nu)^2\equiv \inf \int_0^1\int_{\RR^n} \| V(t,x) \|_{\RR^n}^2 \gamma(t,x)~dt,
\end{align}
where $\|.\|_{\RR^n}$ is the usual Euclidean norm on $\RR^n$, and where the infimum is taken over paths $\gamma:[0,1]\times \RR^n\to\RR$ and vector fields $V:[0,1]\times \RR^{n}\to \RR^n$ that satisfy the following \textit{continuity equation}
\begin{align*}
		\partial_t \gamma(t,x)+\operatorname{div} (\gamma(t,x)V(t,x))=0,~~ \gamma(0,.)=\varphi,~~ \gamma(1,.)=\psi.
\end{align*}	
Here, $\operatorname{div}$ denotes the classical divergence operator on $\RR^n$. This characterization of $W_2$ paved the way for a definition of a non-commutative extension of it by Carlen and Maas (\cite{[CM14],[CM16]}). They defined a \textit{non-commutative Wasserstein distance} $W_{2,\LL}$ as follows: For a quantum Markov semigroup with generator $\LL$ of the form of \reff{LLDBC}, let $\gamma(s)$, $s\in (-\eps,\eps)$ be a differential path in $\cD_+(\cH)$ for some $\eps>0$, and denote $\rho:=\gamma(0)$. Then $\tr(\dot{\gamma}(0))=\frac{d}{ds}\left|_{s=0}\tr(\gamma(s))=0\right.$. Carlen and Maas proved that there is a unique vector field $\mathbf{V}\in\bigoplus_{j\in\mathcal{J}} \cB(\cH) $ of the form $\mathbf{V}=\nabla U$, where $U\in\cB(\cH)$ is traceless and self-adjoint, for which the following non-commutative continuity equation holds:
\begin{align}\dot{\gamma}(0)=-\operatorname{div}([\rho]_{\vec{\omega}}\nabla U),\label{continuity2}
\end{align}
 Define the following inner product $\langle .,. \rangle_{\LL,\rho}$ on $\bigoplus_{j\in\mathcal{J}}\cB(\cH)$:
 \begin{align*}
 	\langle \mathbf{W},\mathbf{V}\rangle_{\LL,\rho}:=\sum_{j\in\mathcal{J}} c_j \langle W_j,[\rho]_{\omega_j}{V}_j\rangle_{_{HS}} ,
 \end{align*}
 where $\langle A,B\rangle_{_{HS}} :=\tr(A^* B)$ denotes the usual Hilbert Schmidt inner product on $\cB(\cH)$. 
 Hence, looking upon $\cD_+(\cH)$ as a manifold, for each $\rho\in \cD_+(\cH)$, we can identify the tangent space $T_\rho$ at $\rho$ with the set of gradient vector fields $\{\nabla U:~U\in\cB(\cH),~ U=U^*\}$ through the correspondence provided by the continuity equation (\ref{continuity2}). Defining the metric $g_\LL$ through the relation
 \begin{align}\label{gammav}
 	\|\dot{\gamma}(0)\|_{g_{\LL,\rho}}^2:= \|\mathbf{V}(s)\|^2_{\LL,\rho},
 \end{align}
 this endows the manifold $\cD_+(\cH)$ with a smooth Riemmanian structure. In this framework, Carlen and Maas then defined the \textit{modified non-commutative Wasserstein distance} $W_{2,\LL}$ to be the energy associated to the metric $g_{\LL}$, i.e.: 
	\begin{align}\label{W2}
 		W_{2,\LL}(\rho,\sigma):=\inf_{\gamma}\left\{ \left(\int_0^1 
 		\|\mathbf{V}(s)\|^2_{\LL,\gamma(s)}
 		ds\right)^{1/2}~:~\gamma(0)=\rho,~~\gamma(1)=\sigma\right\}
 		\end{align}
 		where the infimum is taken over smooth paths $\gamma:[0,1]\to \cD_+(\cH)$, and $\mathbf{V}\equiv \nabla U:[0,1]\to \bigoplus_{i\in\cJ}\cB(\cH)$ is related to $\gamma$ through the continuity equation \reff{continuity2}. The paths achieving the infimum, if they exist, are the minimizing geodesics with respect to the metric $g_\LL$. This expression for the non-commutative Wasserstein distance is to be compared with its continuous classical analogue provided in \Cref{Benambren}
 		\begin{remark}The metric ${g_{\LL}}$ can be related to a modified version of the  \textit{Kubo-Mori-Bogoliubov Fisher information metric} on the set $\cD_+(\cH)$ seen as a manifold (see e.g. \cite{[CGT16]}). Other possible quantum extensions of the Wasserstein distance of order $2$ were considered in \cite{[CGT16]} in the case when the unique invariant state of the quantum Markov semigroup is the completely mixed state, starting from different versions of quantum Fisher informations.
 			\end{remark}
 		\noindent
 			  The Wasserstein distance $W_{2,\LL}$ can be alternatively expressed as follows:
 		\begin{lemma}\label{characwass}
 			With the above notations, the Wasserstein distance $W_{2,\LL}$ between two full-rank states $\rho,\sigma$ is equal to the minimal length over the smooth paths joining $\rho$ and $\sigma$:
 			\begin{align}\label{length}
 				W_{2,\LL}(\rho,\sigma)=\inf_{\gamma\text{ const. speed}}\left\{\int_0^1 \|\dot{\gamma}(s)\|_{g_{\LL,\gamma(s)}}ds:~\gamma(0)=\rho,~\gamma(1)=\sigma\right\},
 				\end{align}
 				where the infimum is taken over the set of curves $\gamma$ of constant speed, i.e. such that $s\mapsto \|\dot{\gamma}(s)\|_{g_{\LL,\gamma(s)}}$ is constant on $[0,1]$. 
 		\end{lemma}
 		\begin{proof}
By the proof of Lemma 1.1.4 (b) of \cite{AGS08}, for any smooth path $\gamma:[0,1]\to \cD_+(\cH)$, with $\gamma(0)=\rho$ and $\gamma(1)=\sigma$ there exists a path $\eta$ with the same boundary conditions and constant speed, i.e.~for any $s\in [0,1]$
\begin{align*}
	\|\dot{\eta}(s)\|_{g_{\LL,\eta(s)}}=\int_0^1 	\|\dot{\gamma}(u)\|_{g_{\LL,\gamma(u)}} du.
	\end{align*}
Fix such a path $\gamma$. By the Cauchy-Schwarz inequality, for any $s\in [0,1]$
 	\begin{align}\label{ineq111}
 	 			\|\dot{\eta}(s)\|_{g_{\LL,\eta(s)}}=	\int_0^1\|\dot{\gamma}(u)\|_{g_{\LL,\gamma(s)}}du\le \left(\int_0^1 \|\dot{\gamma}(u)\|_{g_{\LL,\gamma(s)}}^2du\right)^{1/2}.
 	 				\end{align}
with equality if $\gamma=\eta$. The statement follows by taking the infimum over all smooth paths $\gamma$ on the right hand side of \reff{ineq111}.
 		\qed
 		\end{proof}	
 		\noindent
 			This definition for the quantum Wasserstein distance, $W_{2,\LL}$, is natural in the sense that the path $\gamma(t):=\rho_t$, which is a solution of the master equation
 		\begin{align*}
 			\dot{\rho}_t=\LL_*\rho_t,
 			\end{align*}
 		is the \textit{gradient flow} for $D(.\|\sigma)$ in the metric $g_\LL$, where $\sigma$ is the invariant state associated to $\LL$. This means that $\LL_* \rho=-\operatorname{grad}_{g_\LL} D(\rho\|\sigma)$, where the gradient $\operatorname{grad}_{g_\LL}$ of a differentiable functional $\mathcal{F}:\cD_+(\cH)\to \RR$ is defined as the unique element in the tangent space at $\rho$ such that 
 		\begin{align}\label{gflow}
 		\left.	\frac{d}{dt}\mathcal{F}(\gamma(t))\right|_{t=0}=g_{\LL,\rho}(\dot{\gamma},\operatorname{grad}_{g_{\LL,\rho}}\mathcal{F}(\rho))
\end{align}
for all smooth paths $\gamma(t)$ defined on $(-\eps,\eps)$ for some $\eps>0$ with $\gamma(0)=\rho$, where the quantity on the right hand side of \Cref{gflow} is the metric defined through \Cref{gammav}. In particular, for $\gamma(t)=\rho_t\equiv \Lambda_{*t}(\rho)$,
\begin{align}\label{gradflow}
	\left.\frac{d}{dt}D(\rho_t\|\sigma)\right|_{t=0}=-g_{\LL,\rho}(\LL_*(\rho),\LL_*(\rho))=-\|\LL_*(\rho)\|_{g_{\LL,\rho}}^2.
\end{align}
\subsubsection*{The quantum Wasserstein distance of order 1}	 
The case $p=1$ can also be extended to the quantum setting as follows. The classical Wasserstein distance $W_1(\mu,\nu)$ between two probability measures $\mu,\nu$ has the following dual representation:
\begin{align*}
	W_1(\mu,\nu)=\sup_{\varphi:\,\|\varphi\|_{\operatorname{lip}}\le 1}\left\{\int \varphi d\mu -\int \varphi d\nu\right\},
	\end{align*}
where $\|\varphi\|_{\operatorname{lip}}$ denotes the Lipschitz constant of the function $\varphi$. This representation is due to Kantorovich and Rubinstein \cite{[V08]}, and for this reason is also called the \textit{Kantorovich-Rubinstein distance}. In the same spirit, given a quantum Markov semigroup $(\Lambda_t)_{t\ge0}$ of the form of \Cref{LLDBC} on a Hilbert space $\cH$ of dimension $d$, we define, by analogy with the classical case, the following Lipschitz constant of a self-adjoint operator $f$ :
	\begin{align}\label{lipnorm}
	\|f\|_{\operatorname{Lip}}:= \left(\frac{1}{d}  \sum_{j\in\mathcal{J}} c_j (\e^{-\omega_j/2}+\e^{\omega_j/2})\|\partial_jf\|_{\infty}^2\right)^{1/2}.
	\end{align}
 Next we define the \textit{non-commutative $1$-Wasserstein distance} between two states $\rho$ and $\sigma$ associated to a semigroup with generator $\LL$ of the form of \reff{LLDBC} to be
	\begin{align*}
		W_{1,\LL}(\rho,\sigma):=\sup_{f\in\cB_{sa}(\cH):\,\|f\|_{\operatorname{Lip}}\le 1}|\tr( f(\rho-\sigma))|.
		\end{align*}
		\begin{lemma}\label{wasslemm}
			The non-commutative $1$-Wasserstein distance $W_{1,\LL}$ defines a distance on $\cD_+(\cH)$.
		\end{lemma}
	\begin{proof}
Symmetry and non-negativity are obvious by definition. If $W_{1,\LL}(\rho,\sigma)=0$, then for all $f\in\cB_{sa}(\cH)$, $\tr (f(\rho-\sigma))=\langle f,\rho-\sigma\rangle_{_{HS}}=0$, which implies $\rho=\sigma$. Finally, let $\rho,\sigma,\tau\in\cD_+(\cH)$. Then, for any $f\in\cB_{sa}(\cH)$ such that $\|f\|_{\operatorname{Lip}}\le 1$, we have by the triangle inequality
\begin{align*}
|\tr (f(\rho-\tau))|\le |\tr f(\rho-\sigma)|+|\tr f(\sigma-\tau)|\le W_{1,\LL}(\rho,\tau)+W_{1,\LL}(\sigma,\tau).
\end{align*}
The result follows by taking the supremum over such operators $f$ on the right hand side of the above inequality.
	\qed
	\end{proof}
	\noindent
	In the classical setting, a simple use of H\"{o}lder's inequality leads to the fact that 
 \begin{align*}
 1\le	p\le q<\infty~~~\Rightarrow~~~ W_p\le W_q.
 	\end{align*}
  	The following lemma shows that this hierarchy is maintained, up to a dimensional factor, in the quantum case:
 	\begin{lemma}\label{W1W2}
 		For any $\rho,\sigma\in\cD_+(\cH)$, 
 		\begin{align*}
 			W_{1,\LL}(\rho,\sigma)\le \sqrt{d}~W_{2,{\LL}}(\rho,\sigma),
\end{align*}
\end{lemma}
\begin{proof}
	The proof is inspired by the proof of Proposition 2.12 of \cite{[EM12]}. Fix $\delta>0$, and $\rho,\sigma\in\cD_+(\cH)$. There exists a smooth path $(\gamma(s))_{s\in [0,1]}$ in $\cD_+(\cH)$ such that $\gamma(0)=\rho$, $\gamma(1)=\sigma$, and by definition of $W_{2,\LL}$:
		\begin{align}\label{upbound}
		\left(\int_{0}^1\sum_{j\in\mathcal{J}} c_j \langle V(s)_j,[\gamma(s)]_{\omega_j}{V(s)}_j\rangle_{_{HS}}\right)^{1/2} \equiv	\left( \int_0^1 \|\mathbf{V}(s)\|_{\LL,\gamma(s)}^2\right)^{1/2}\le  W_{2,\LL}(\rho,\sigma)+\delta.
	\end{align}
	Recall that for each $s\in[0,1]$, $\mathbf{V}(s)\equiv(V(s)_1,...,V(s)_{|\mathcal{J}|})$ is related to $\gamma(s)$ through the continuity equation:
	\begin{align*}
				\dot{\gamma}(s)+\operatorname{div}([\gamma(s)]_{\vec{\omega}}\mathbf{V}(s))=0,
	\end{align*}	
where the divergence operator $\operatorname{div}$ is defined in \Cref{div}. Hence, for any $f\in \cB_{sa}(\cH)$ such that $\|f\|_{\operatorname{Lip}}\le 1$,
	\begin{align}
		|\tr(f(\rho-\sigma))|&=\left|\tr\left( f\int_0^1\frac{d}{ds}\gamma(s) ds\right)\right|\nonumber\\
		&=\left| \int_0^1 \tr \left(f \operatorname{div}([\gamma(s)]_{\vec{\omega}}\mathbf{V}(s))\right)\right|\nonumber\\
		&=\left| \int_0^1\sum_{j\in \mathcal{J}} c_j \langle \partial_j f, [\gamma(s)]_{\omega_j} {V}(s)_j\rangle_{_{HS}} ds\right|\nonumber\\
		&\le \left(\int_0^1 \sum_{j\in\mathcal{J}} c_j\langle \partial_j f,[\gamma(s)]_{\omega_j}\partial_jf\rangle_{_{HS}}  ds\right)^{1/2} \left(\int_0^1 \sum_{j\in\mathcal{J}} c_j\langle V(s)_j,[\gamma(s)]_{\omega_j}V(s)_j\rangle_{_{HS}}  ds\right)^{1/2}\nonumber\\
		&\le \left(\int_0^1 \sum_{j\in\mathcal{J}} c_j\langle \partial_j f,[\gamma(s)]_{\omega_j}\partial_jf\rangle_{_{HS}}  ds\right)^{1/2} (W_{2,\LL}(\rho,\sigma)+\delta)\label{eq13}
			\end{align}
	where in the fourth line we used the Cauchy-Schwarz inequality with respect to the inner product $\sum_{j\in\cJ} c_j \int_0^1 \langle.,[\gamma(s)]_{\omega_j}.\rangle_{_{HS}} ds$, and the last line comes from \reff{upbound}. Now, for each $j\in\mathcal{J}$, and any $s\in [0,1]$, it can be shown that
	\begin{align}\label{eq10}
		[\gamma(s)]_{\omega_j}\le \e^{-\omega_j/2}R_{\gamma(s)}+\e^{\omega_j/2}L_{\gamma(s)},
	\end{align}	
where $L_{\gamma(s)}$ and $R_{\gamma(s)}$ are the operators of left multiplication, respectively right multiplication, by $\gamma(s)$.	This inequality follows from the fact that the metric $g_{\LL}$ involves a sum of Kubo-Mori-Bogoliubov quantum Fisher information metrics, as introduced by Petz in \cite{P96} (see also e.g. \cite{LR99,HK99,PWPR06,HKPR13,TKRWV10}). Indeed, from \Cref{fomega}, we have
\begin{align}
 (f^{-1}_{\omega_j})(\Delta_{\gamma(s)})&=\e^{-\omega_j/2}({\log \Delta_{\gamma(s)}+\omega_j})({\Delta_{\gamma(s)}-\e^{-\omega_j}})^{-1}\nonumber\\
	&=\e^{\omega_j/2}({\log \Delta_{\e^{\omega_j}\gamma(s)|\gamma(s)}})({\Delta_{\e^{\omega_j}\gamma(s)|\gamma(s)}-1})^{-1}\nonumber\\
	&\equiv\e^{\omega_j/2}k^{\log}(\Delta_{\e^{\omega_j}\gamma(s)|\gamma(s)}),\label{eq11}
\end{align}		
	where, using the notations of \cite{LR99}, $k^{\log}:\RR_+\ni t\mapsto (t-1)^{-1}\log(t)$ is such that $-k^{\log}$ is operator monotone, $k^{\log}(t^{-1})=tk^{\log}(t)$, and $k^{\log}(1)=1$. Therefore, from Theorem 2.11 of \cite{LR99}, it follows that
	\begin{align}\label{eq12}
		R_{\gamma(s)}^{-1} k^{\log}(\Delta_{\e^{\omega_j}\gamma(s)|\gamma(s)})\ge (R_{\gamma(s)}+L_{\e^{\omega_j}\gamma(s)})^{-1}.
		\end{align}
	The inequality \reff{eq10} follows by substituting \reff{eq11} into \reff{eq12}. Therefore,
	\begin{align}
		 \langle \partial_j f,[\gamma(s)]_{\omega_j}\partial_j f\rangle_{_{HS}}  &\le \e^{-\omega_j/2} \tr(\gamma(s)\partial_jf^*\partial_jf)+\e^{\omega_j/2}\tr(\partial_jf^* \gamma(s)\partial_j f)\nonumber\\
		 &\le (\e^{-\omega_j/2} +\e^{\omega_j/2})\|\partial_jf\|_{\infty}^2,\label{lost}
		 \end{align} 
		 where the last line follows by H\"{o}lder's inequality. Substituting this result into \reff{eq13}, we end up with
		 \begin{align}
	|\tr(f(\rho-\sigma))|\le \sqrt{d}~\|f\|_{\operatorname{Lip}}(W_{2,\LL}(\rho,\sigma)+\delta)\le \sqrt{d}~(W_{2,\LL}(\rho,\sigma)+\delta),
	\end{align}
	since by assumption $\|f\|_{\operatorname{Lip}}\le1$. The result follows by taking the limit $\delta\to 0$, and optimizing the left hand side of the above inequality over all such $f$.
	\qed
	\end{proof}
\noindent
		Notice that the expression \reff{lipnorm} depends on the representation \reff{LLDBC}. An alternative definition of the quantum Lipschitz constant which is independent of the representation is as follows:
		\begin{align*}
		\|f\|_{\operatorname{Lip,2}}:=  \left(\frac{1}{d}  \sum_{j\in\mathcal{J}} c_j (\e^{-\omega_j/2}+\e^{\omega_j/2})\|\partial_jf\|_{2}^2\right)^{1/2}.
	\end{align*}	
 $\|f\|_{\operatorname{Lip},2}$ does not depend on the representation $(\{c_j\},\{\tilde{L_j}\},\{\omega_j\})$ in \reff{LLDBC}. Indeed, given two such representations $(\{c_j\},\{\tilde{L_j}\},\{\omega_j\})$ and $(\{c_j'\},\{\tilde{L_j'}\},\{\omega_j'\})$,
 \begin{align*}
 	 \sum_{j\in\mathcal{J}} c'_j (\e^{-\omega_j'/2}+\e^{\omega_j'/2})\|\partial_j'f\|_{2}^2&= \sum_{j\in\cJ}c_j' (\e^{-\omega_j'/2}+\e^{\omega_j'/2})\tr([f,{\tilde{L'}_j^*}][\tilde{L}_j',f])\\
 	 &=\sum_{j,l,k\in\cJ}c_j' (\e^{-\omega_j'/2}+\e^{\omega_j'/2})U_{jk}\overline{U}_{jl}\tr([f,{\tilde{L}_l^*}][\tilde{L}_k,f]   )\\
 	 &=\sum_{j,l,k\in\cJ}U_{jk}\overline{U}_{jl}c_j' (\e^{-\omega_k/2}+\e^{\omega_k/2})\tr([f,{\tilde{L}_l^*}][\tilde{L}_k,f]   )\\
 	 &=\sum_{l,k\in\cJ}c_l\delta_{kl} (\e^{-\omega_k/2}+\e^{\omega_k/2})\tr([f,{\tilde{L}_l^*}][\tilde{L}_k,f]   )\\
 	&= \sum_{l\in\cJ}c_l (\e^{-\omega_l/2}+\e^{\omega_l/2})\tr([f,{\tilde{L}_l^*}][\tilde{L}_l,f]   )\\
 	 &=\sum_{j\in\cJ} c_j (\e^{-\omega_j/2}+\e^{\omega_j/2})\|\partial_jf\|_{2}^2,
 	  \end{align*}
  where we used \Cref{rmkU} in the above lines. Moreover, from $\|f\|_\infty\le \|f\|_2$, it follows that $\|f\|_{\operatorname{Lip}}\le \|f\|_{\operatorname{Lip},2}$. Two other Lipschitz constants can be defined:
\begin{align*}
&\|f\|_{\operatorname{Lip,g}}:=\sup_{j,~c_j\ne 0} \|\partial_j f\|_{2} \left(\frac{2}{d}\sum_{k\in \cJ} c_k \e^{-\omega_k/2}  \right)^{1/2},\\
&\|f\|_{\operatorname{Lip,H}}:=\sup_{j} \|\partial_j f\|_{2} \left(\frac{2}{d}\sum_{k\in \cJ} c_k \e^{-\omega_k/2} \right)^{1/2},
\end{align*}
where we used that for any $k\in\cJ$, there exists $k'\in\cJ$ such that $\omega_k=-\omega_{k'}$ and $c_k=c_{k'}$ (cf. \Cref{LLDBC}). For each of these Lipschitz constants $\|f\|_{\operatorname{Lip},\star}$, one can define an associated quantum Wasserstein distance of order $1$:
\begin{align}\label{WALT}
W_{1,\LL,\star}(\rho,\sigma):=\sup_{f\in\cB_{sa}(\cH), ~\|f\|_{\operatorname{Lip},\star}\le 1}|\tr f(\rho-\sigma)|
\end{align}
for which \Cref{wasslemm} extends, i.e. $W_{1,\LL,\star}$ defines a distance on $\cD_+(\cH)$. The following proposition follows directly from the direct observation that 
\begin{align*}
\|f\|_{\operatorname{Lip}}\le \|f\|_{\operatorname{Lip,2}}\le \|f\|_{\operatorname{Lip,g}}\le \|f\|_{\operatorname{Lip,H}}.
\end{align*}
\begin{proposition}\label{proplip}
	For any $\rho,\sigma\in\cD_+(\cH)$,
	\begin{align*}
	W_{1,\LL,\operatorname{H}}(\rho,\sigma)\le W_{1,\LL,\operatorname{g}}(\rho,\sigma)\le W_{1,\LL,2}(\rho,\sigma)\le W_{1,\LL}(\rho,\sigma)\le \sqrt{d}~ W_{2,\LL}(\rho,\sigma).
	\end{align*}
\end{proposition}	
\noindent
The reason behind the use of the denominations $\operatorname{g}$ and $\operatorname{H}$ comes from the fact that the associated Lipschitz constants reduce to the classical Lipschitz constants associated to the graph and Hamming distances respectively as follows: 	Let $\Omega=\{1,...,d\}$ be a finite classical state space, and fix an orthonormal basis $\{|i\rangle\}_{i=1,...,d}$ of vectors associated to each element of $\Omega$. Moreover, define the quantum Markov semigroup of generator $\LL$ of the form given in \Cref{LLDBC} with $\tilde{L}_{j}\equiv \tilde{L}_{kl}=\sqrt{d}~|k\rangle\langle l|$, where we assume here that $\cJ$ is a finite set of pairs $j=(k,l)$. Then, define 
		\begin{align*}
		\LL(f)=\sum_{k\ne l}   ~    d~c_{k,l}~\e^{-\omega_{k,l}/2}(~ |l\rangle\langle k|~[f,|k\rangle\langle l|]   ~+~[|l\rangle\langle k|,f] ~|k\rangle\langle l|~).
		\end{align*}
		Assume further that $f=\sum_{k=1}^d \varphi(k)|k\rangle\langle k|$, for some real valued function $\varphi:\Omega\mapsto \RR$. In this case (see Theorem 4.2 of \cite{[CM16]}), the quantum Markov semigroup of generator $\LL$ induces a classical Markov semigroup on $\Omega$ with associated $Q$ matrix \cite{norris1998markov}
		\begin{align*}
		Q_{k,l}:=2d\sum_{k'\ne l'} c_{k',l'}\e^{-\omega_{k',l'}/2} \langle k| k'\rangle \langle l'|l\rangle\langle l|l'\rangle\langle l'|k\rangle= 2~d~c_{k,l}~\e^{-\omega_{k,l}/2},~~~~k\ne l
		\end{align*}
		and so that $\sum_{l\in\Omega}Q_{k,l}=0$. Then,
		\begin{align*}
		\LL(f)=\sum_{k\in\Omega} (Q\varphi)(k)~|k\rangle\langle k|.
		\end{align*}
		Assuming further that, for each $k=1,...,d$, $a_k:=\sum_{l\ne k} Q_{k,l}=1$, define the transition matrix
		\begin{align*}
		P_{k,l}:=\left\{\begin{aligned}
		&Q_{k,l}~~~~~~~~~~l\ne k,\\
		&0~~~~~~~~~~~~~~l=k.
		\end{aligned}\right.
		\end{align*}
Therefore, the expressions of the Lipschitz constants $\|f\|_{\operatorname{Lip},\operatorname{g}}$ and $\|f\|_{\operatorname{Lip},\operatorname{H}}$ reduce to
		\begin{align*}
\|f\|_{\operatorname{Lip,g}}= \sup_{c_{k,l}\ne 0}   \|[|k\rangle\langle l|,f ]\|_2 \left(\frac{1}{d}\sum_{k'\ne l'}P_{k',l'}\right)^{1/2}=\sup_{P_{k,l}\ne 0}|\varphi(k)-\varphi(l)|.
		\end{align*}
This is exactly the Lipschitz constant
\begin{align*}
\|\varphi\|_{\operatorname{lip},\operatorname{d_g}}:=\sup_{k\ne l}\frac{|\varphi(k)-\varphi(l)|}{\operatorname{d_g}(k,l)},
\end{align*}
associated to the graph distance 
\begin{align*}
 \operatorname{d}_{\operatorname{g}}:=\left\{ 
\begin{aligned}
&1~~~~~~~~~~ c_{k,l}\ne 0\\
&0~~~~~~~~~~c_{k,l}=0.
\end{aligned}\right.
\end{align*}
Similarly, the Lipschitz constant $\|f\|_{\operatorname{Lip,H}}$ reduces to the Lipschitz constant
\begin{align*}
\|\varphi\|_{\operatorname{lip},\operatorname{d_H}}:=\sup_{k\ne l}\frac{|\varphi(k)-\varphi(l)|}{\operatorname{d_H}(k,l)},
\end{align*}
associated to the Hamming distance 
\begin{align*}
\operatorname{d}_{\operatorname{H}}:=\left\{ 
\begin{aligned}
&1~~~~~~~~~~ {k\ne l},\\
&0~~~~~~~~~~k=l.
\end{aligned}\right.
\end{align*}
Hence, \Cref{proplip} extends Proposition 2.12 of \cite{[EM12]}. Note that the bound obtained in \Cref{W1W2} is weaker than its classical counterpart $\sqrt{2}~W_1(p,q)\le W_2(p,q)$, for $p,q$ positive probability vectors. The reason behind this weaker bound comes from the very last line \reff{lost} of the proof of Lemma 6.
\section{Quantum functional- and transportation cost inequalities}\label{qfunc}
In this section, we consider a primitive quantum Markov semigroup $(\Lambda_t)_{t\ge 0}$ on $\cB(\cH)$ with invariant state $\sigma\in\cD_+(\cH)$, whose generator is of the form given in \reff{LLDBC}. Such a semigroup is said to satisfy the following:
\begin{itemize}
	\item[1.] A modified logarithmic Sobolev (or $1$-log-Sobolev) inequality with constant $\alpha_1>0$ if for all $f\in \cP(\cH)$,
	\begin{align}
		\tag{MLSI($\alpha_1$)}
		\alpha_1 \operatorname{Ent}_{1,\sigma}(f)\le \mathcal{E}_{1}(f,f),\label{ls1}
		\end{align}
		or equivalently if for all $\rho\in\cD_+(\cH)$,
		\begin{align}\label{ls11}
		2	\alpha_1 D(\rho\|\sigma)\le  \operatorname{EP}_\sigma(\rho)=\operatorname{I}_\sigma(\rho).
			\end{align}
				\item[2.] A transportation cost inequality of order $1$ with constant $c_1>0$ if for all $\rho\in \cD_+(\cH)$
				\begin{align}
					\tag{TC$_1$($c_1$)}\label{t1}
					W_{1,\LL}(\rho,\sigma)\le \sqrt{2c_1D(\rho\|\sigma)}.
				\end{align}
		\item[3.] A transportation cost inequality of order $2$ with constant $c_2>0$ if for all $\rho\in \cD_+(\cH)$
		\begin{align}
			\tag{TC$_2$($c_2$)}\label{t2}
			W_{2,\LL}(\rho,\sigma)\le \sqrt{2c_2D(\rho\|\sigma)}.
			\end{align}
			\item[4.] A Poincar\'{e} inequality with constant $\lambda>0$, with respect to $\langle .,.\rangle_{\varphi,\sigma}$, where $\varphi:(0,\infty)\to (0,\infty)$,  if for all $f\in\cB(\cH)$ with $\tr(\sigma f)=0$,
			\begin{align}\label{p}
				\tag{PI($\lambda$)}
			\lambda	\|f\|_{\varphi,\sigma}^2\le \mathcal{E}_{\varphi,2}(f,f).
			\end{align}
\end{itemize}
Classically, it is proved that (\ref{ls1})$\Rightarrow$(\ref{t2})$\Rightarrow$(\ref{p}), for $\alpha_1\le 1/{c_2}\le\lambda$, as well as (\ref{t2})$\Rightarrow$(\ref{t1}) for $c_2\ge c_1$. Moreover, (\ref{ls1}) and (\ref{t2}) imply Gaussian concentration, whereas (\ref{p}) only implies exponential concentration (see \cite{[V08],[RS14],[EM12]} and references therein). Here we state and prove the non-commutative analogues of these results. Note that the implication \reff{ls1}$\Rightarrow$(\ref{p}) for $\alpha_1=\lambda$ was proved in \cite{[KT13]} in the quantum framework for primitive semigroups which are self-adjoint with respect to $\langle .,.\rangle_{1/2,\sigma}$. Finally, the implication \reff{ls1}$ \Rightarrow$\reff{t2} was already proved in \cite{[CM14]} is the particular case of the Fermionic Fokker-Planck equation.
\begin{remark}
		Even though we defined several quantum Lipschitz constants above, among which $\|f\|_{\operatorname{Lip},2}$ does not depend on the particular choice of representation $(\{c_j\},\{\tilde{L}_j\},\{\omega_j\})$ in \Cref{LLDBC}, we only work with $\|f\|_{\operatorname{Lip}}$, unless otherwise stated. This is because, by \Cref{proplip}, the transportation cost inequality \reff{t1} for the associated Wasserstein distance $W_{1,\LL}$ implies any transportation cost inequality with respect to the other Wasserstein distances $W_{1,\LL,*}$.
		\end{remark}
\begin{remark}\label{onpoincare}
	If we assume $\sigma$ non-degenerate and $\LL$ to be self-adjoint with respect to $\langle.,. \rangle_{1,\sigma}$, we know from \Cref{selfadjointness} that it is also self-adjoint with respect to $\langle.,.\rangle_{\varphi,\sigma}$ for any $\varphi:(0,\infty)\to (0,\infty)$. Therefore $\LL$ admits a spectral decomposition: 
	\begin{align*}
		\LL=\sum_{\lambda\in \spec(\LL)} \lambda P_\lambda(\LL),
		\end{align*}
		 However, since for each $t\ge 0$, $\Lambda_t:=\e^{t\LL}$, one can similarly write
		\begin{align*}
			\Lambda_t=\sum_{\lambda\in \spec(\LL)} \e^{t\lambda} P_\lambda(\LL).
			\end{align*}
		The eigenvalues of $\LL$ are non-positive, and the highest one $\lambda_0=0$ is associated to the eigenspace $\CC \mathbb{I}$. Therefore, the inequality (\ref{p}) holds for any $\lambda<\lambda_1$ and any $\varphi:(0,\infty)\to(0,\infty)$, where the so-called \textit{spectral gap} $\lambda_1$ is nothing but the absolute value of the second (negative) highest eigenvalue of $\LL$.
\end{remark}
\noindent
The following result, is a direct consequence of \Cref{W1W2}. 
\begin{theorem}\label{t2t1}		
	If $(\Lambda_t)_{t\ge0}$ satisfies the $2$-transportation cost inequality (\ref{t2}), then it satisfies the $1$-transportation cost inequality \reff{t1} with $c_1=d ~c_2$.
\end{theorem}	
\noindent
The following result, namely that the modified log-Sobolev inequality \reff{ls1} implies the transportation cost inequality (\ref{t2}) was first proven in the classical, continuous case by Otto and Villani in \cite{[OV00]} (see also \cite{[BGL01],[G09]} for alternative proofs, Theorem 7.5 of \cite{[EM12]} for the classical, discrete case, and \cite{[CM14]} for the case of the fermionic Fokker-Planck semigroup).
\begin{theorem}\label{logsobtalagrand}
If $(\Lambda_t)_{t\ge0}$ satisfies \reff{ls1}, then \reff{t2} holds with $c_2=\alpha_1^{-1}$.
\end{theorem}
\begin{proof} Here we adapt the proof of \cite{[EM12]} to the case of full-rank quantum states. We first state and prove the following lemma:
	\begin{lemma}\label{lemmauseful}
		Let $\rho,\tau\in\cD_+(\cH)$. Then for all $t>0$, $\rho_t\equiv \Lambda_{*t}(\rho)$ satisfies
		\begin{align*}
			\frac{d}{dt}W_{2,\LL}(\rho_t,\tau)\le \sqrt{\operatorname{I}_\sigma(\rho_t)},
		\end{align*} 
	where $\operatorname{I}_\sigma$ is the Fisher information defined in \Cref{Fisher}.
	\end{lemma}
	\begin{proof}
		We proceed here similarly to the proof of Proposition 7.1 of \cite{[EM12]}: Firstly, by the triangle inequality,
		\begin{align}\label{eq8}
			\frac{d}{dt}W_{2,\LL}(\rho_t,\tau)=\lim_{s\to0}\frac{1}{s}(W_{2,\LL}(\rho_{t+s},\tau)-W_{2,\LL}(\rho_t,\tau))\le \lim_{s\to0}\frac{1}{s}W_{2,\LL}(\rho_t,\rho_{t+s}),
		\end{align}
Now, by \Cref{characwass},
			\begin{align*}
				W_{2,\LL}(\rho_t,\rho_{t+s})=\inf_{\gamma(s)}\left\{\int_0^1 \|\dot{\gamma}(u)\|_{g_{\LL,\gamma(u)}}du:~\gamma(0)=\rho_t,~\gamma(1)=\rho_{t+s}\right\},
			\end{align*}
			 This implies by a change of variable $v=t+us$ that for any smooth curve $\gamma$ such that $\gamma(t)=\rho_t$ and $\gamma(t+s)=\rho_{t+s}$,
			\begin{align}\label{eq29}
				W_{2.\LL}(\rho_t,\rho_{t+s})\le \int_{t}^{t+s}\|\dot{\gamma}(v)\|_{g_{\LL,\gamma(v)}}dv.
			\end{align}
			 Moreover, from \Cref{gradflow}:
			 \begin{align}\label{chain}
			 	\|\dot{\rho}_t\|^2_{g_{\LL,\rho_t}}&=-\frac{d}{dt}D(\rho_t\|\sigma)=\operatorname{I}_\sigma(\rho_t),
			 \end{align}
		 		where the second identity holds by de Bruijn's identity (\ref{debruijn}). Hence, choosing $\gamma(v)=\rho_v$, we bound the right hand side of \reff{eq8} as follows:
		\begin{align*}
			\frac{d}{dt}W_{2,\LL}(\rho_t,\tau)\le \lim_{s\to 0} \frac{1}{s}\int_{t}^{t+s} \sqrt{\operatorname{I}_\sigma(\rho_v)}dv=\sqrt{\operatorname{I}_\sigma(\rho_t)},
		\end{align*}
		where the last equality holds since $t\to \sqrt{\operatorname{I}_\sigma(\rho_t)}$ is continuous.
		\qed
	\end{proof}
\bigskip
\noindent
	We now proceed with the proof of \Cref{logsobtalagrand}: Fix $\rho\in\cD_+(\cH)$, and set $\rho_t=\Lambda_{*t}(\rho)$. First note that as $t\to\infty$, 
	\begin{align}
		D(\rho_t\|\sigma)\to0 \text{ and }W_{2,\LL}(\rho,\rho_t)\to W_{2,\LL}(\rho,\sigma) \label{conv}
	\end{align}
	Define now the function 
	\begin{align*}
		F(t):=W_{2,\LL}(\rho_t,\rho)+\sqrt{\frac{2}{\alpha_1}D(\rho_t\|\sigma)}.
	\end{align*}
	Obviously $F(0)=\sqrt{2D(\rho\|\sigma)/\alpha_1}$, and by \reff{conv}, $F(t)\to W_{2,\LL}(\sigma,\rho)$ as $t\to\infty$. Hence it is sufficient to prove that $F$ is non-increasing. In order to do so, we only need to show that its derivative is non-positive. If $\rho_t\ne \sigma$, we know from \Cref{lemmauseful} that 
	\begin{align*}
		\frac{d}{dt}F(t)\le \sqrt{\operatorname{I}_\sigma(\rho_t)}+\sqrt{\frac{2}{\alpha_1}}\frac{\frac{d}{dt}D(\rho_t\|\sigma)}{2\sqrt{D(\rho_t\|\sigma)}}=\sqrt{\operatorname{I}_\sigma(\rho_t)}-\frac{\operatorname{I}_\sigma(\rho_t)}{\sqrt{2\alpha_1 D(\rho_t\|\sigma)}}\le 0.
	\end{align*}
	where we used \Cref{ls1} in the last inequality. If $\rho_t=\sigma$, then the relation also holds true, since this implies that $\rho_r=\sigma$ for all $r\ge t$.
	\qed
\end{proof}	
\noindent
\begin{remark}
	In the last version of \cite{[CM16]}, the authors independently added a slightly different proof of \Cref{logsobtalagrand}.
\end{remark}	
\noindent
We show that a refinement of \Cref{W1W2} as well as \Cref{logsobtalagrand} can be used to provide a new proof of the quantum Pinsker inequality.
\begin{theorem}[Quantum Pinsker's inequality] \label{pinsker}
	For any $\rho,\sigma\in\cD_+(\cH)$,
	\begin{align*}
	\|\rho-\sigma\|_1\le \sqrt{2D(\rho\|\sigma)}.
	\end{align*}
	\end{theorem}
\begin{proof}
Let $\LL_{\mathbb{I}/d}$ be the generator of the quantum depolarizing semigroup with unique invariant state $\mathbb{I}/d$:
\begin{align*}
\LL_{\mathbb{I}/d}(f)=\frac{1}{d}\tr(f)\mathbb{I}-f,~~~~~~~~~~ f\in\cB(\cH).
\end{align*}
It is shown in \Cref{depo} of \Cref{qconc} that $\LL_{\mathbb{I}/d}$ can take the following form:
\begin{align*}
\LL_{\mathbb{I}/d}(f)=\frac{1}{2d}\sum_{k, l=1}^d   ~  |k\rangle\langle l|~[~f,~  |l\rangle\langle k|~]   +[~|k\rangle\langle l|,~f~] ~|l\rangle\langle k|,~~~~~~~~~~ f\in\cB(\cH),
\end{align*}
for any orthonormal basis $\{|i\rangle\}$. Recall the proof of \Cref{W1W2} until its last line \reff{lost}, where we showed that for any $\delta>0$, any smooth path $(\gamma(s))_{s\in [0,1]}$ such that $\gamma(0)=\rho$, $\gamma(1)=\sigma$, and $$\left(\int_0^1 \|\dot{\gamma}(s)\|_{g_{\LL,\gamma(s)}}\right)^{1/2}\le W_{2,\LL}(\rho,\sigma)+\delta,$$ and for any self-adjoint operator $f$,
\begin{align}\label{eq20}
|\tr (f(\rho-\sigma))|\le  \left( \int_0^1 \sum_{j\in\mathcal{J}} c_{j} (\tr(\gamma(s)(\partial_jf)^*\partial_jf)+\tr(\gamma(s)\partial_jf(\partial_jf)^*)) ds\right)^{1/2} (W_{2,\LL}(\rho,\sigma)+\delta),
\end{align}
where for the depolarizing semigroup, the index $j\in\cJ$ represents a couple $(k,l)$, so that $\tilde{L}_{kl}=\sqrt{d}~|k\rangle\langle l|$, $c_{kl}=\frac{1}{2d^2}$, for any given orthonormal basis $\{|k\rangle\}_{k=1}^d$. One can verify that, in this case, choosing the basis $\{|k\rangle\}_{k=1}^d$ to be the one diagonalizing the operator $f:=\sum_{k=1}^d \varphi(k)|k\rangle\langle k|$, the term in brackets on the right hand side of \reff{eq20} reduces to $\|\varphi\|_{\operatorname{lip,H}}^2$, so that, letting $\delta$ tend to $0$,
\begin{align*}
|\tr(f(\rho-\sigma))|\le \|\varphi\|_{\operatorname{lip,H}} ~W_{2,\LL}(\rho,\sigma).
\end{align*}
Assuming, moreover, that $0\le f\le \mathbb{I}$, this implies that, for any $k\ne l$, $|\varphi(k)-\varphi(l)|\le 1$, and thus $\|\varphi\|_{\operatorname{lip,H}}\le 1$. By duality,
\begin{align*}
\|\rho-\sigma\|_1\equiv \sup_{0\le f\le \mathbb{I}}|\tr f(\rho-\sigma)|\le \sup_{f=\sum_{j}\varphi(j)|j\rangle\langle j|:~\|\varphi\|_{\operatorname{lip,H}}\le 1} |\tr f(\rho-\sigma)|\le W_{2,\LL}(\rho,\sigma).
\end{align*}
We conclude using \Cref{logsobtalagrand} as well as the well-known fact that in the case of the depolarizing semigroup, $\alpha_1=1$ (see e.g. Lemma 25 of \cite{[KT13]}).
\qed
\end{proof}
\noindent
In \cite{[OV00]} it was also proved that, in the classical, continuous case, the $2$-transportation cost inequality implies the Poincar\'{e} inequality. In the discrete setting, this was proved in Proposition 7.6 of \cite{[EM12]}. \Cref{talagrandpoincare} below extends these results to the quantum regime.
\begin{theorem}\label{talagrandpoincare}
If $(\Lambda_t)_{t\ge0}$ satisfies \reff{t2}, then \reff{p} holds with respect to $\langle .,. \rangle_{1/2,\sigma}$, with $\lambda=(c_2~\kappa_\LL)^{-1}$, where $\kappa_\LL=\sup_{j\in\mathcal{J}}\|[\sigma]_{\omega_j}\circ [\sigma]_{-\omega_j}^{-1}\|_{\infty\to \infty}$.
\end{theorem}
\begin{proof}
Let $f\in\cB(\cH)$ such that $\tr(\sigma f)=0$, and for some $\eps$ small enough, define $f^\eps:=\mathbb{I}+\eps f>0$. Then, define the completely positive, trace-preserving map $\Xi_\sigma$ through the following equation: for any $A\in\cB(\cH)$,
\begin{align*}
	\Xi_\sigma (A):=\int_0^\infty \frac{\sigma^{1/2}}{t+\sigma}A\frac{\sigma^{1/2}}{t+\sigma}dt
\end{align*}
In order to get the result we will need the following two technical lemmas:
		\begin{lemma}\label{lemma8} With the notations of \Cref{LLDBC},
			\begin{align}
			&\tilde{L}_j(t+\sigma)^{-1}=\frac{\e^{-\omega_j}\sigma^{-1}}{1+t\e^{-\omega_j}\sigma^{-1}}\tilde{L}_j\label{eq3}\\
			&(t+\sigma)^{-1} \tilde{L}_j=\tilde{L}_j\frac{\sigma^{-1}\e^{\omega_j}}{1+t\sigma^{-1}\e^{\omega_j}}.\label{eq4}
			\end{align}
		\end{lemma}	
		\begin{proof}
			We first prove \Cref{eq3}:
			\begin{align*}
				\tilde{L}_j(t+\sigma)^{-1}=\e^{-\omega_j}\sigma^{-1} \tilde{L}_j(t+\sigma)(t+\sigma)^{-1}-t\e^{-\omega_j}\sigma^{-1}\tilde{L}_j(t+\sigma)^{-1},
				\end{align*}
				where we used that $\tilde{L}_j$ is an eigenvector of $\Delta_\sigma$ with associated eigenvalue $\e^{-\omega_j}$. Therefore, 
				\begin{align*}
					(1+t\e^{-\omega_j}\sigma^{-1})\tilde{L}_j(t+\sigma)^{-1}=
					\e^{-\omega_j}\sigma^{-1}\tilde{L}_j
	\end{align*}
	and the result follows. Similarly for \Cref{eq4},
	\begin{align*}
		(t+\sigma)^{-1}\tilde{L}_j =(t+\sigma)^{-1}(t+\sigma) \tilde{L}_j\sigma^{-1}\e^{\omega_j}-(t+\sigma)^{-1}t\tilde{L}_j\sigma^{-1}\e^{\omega_j}.
		\end{align*}
		The result again follows by rearranging the above terms.
						\qed
						\end{proof}
						\begin{lemma} For $\Gamma_\sigma(f)\equiv \sigma^{1/2}f\sigma^{1/2}$,
							\begin{align}
								\partial_j(\Xi_\sigma(f))=[\sigma]_{-\omega_j}^{-1}\circ \Gamma_\sigma\circ\partial_jf.\label{eq5}
								\end{align}
										\end{lemma}
		\begin{proof}
			Start from the left hand side of \Cref{eq5}. Using that for each $j\in\cJ$,  $\Delta_{\sigma}^{\pm 1/2}(\tilde{L}_j)=\e^{\mp 1/2}\tilde{L}_j$,
			\begin{align}
				\partial_j(\Xi_\sigma(f))=~&\e^{\omega_j/2}\sigma^{1/2} \tilde{L}_j\int_0^\infty (t+\sigma)^{-1}f(t+\sigma)^{-1}\sigma^{1/2}dt\label{eq6}\\
				& -\e^{-\omega_j/2}\sigma^{1/2}\int_0^\infty (t+\sigma)^{-1}f(t+\sigma)^{-1}\tilde{L}_j \sigma^{1/2}dt\nonumber
				 \end{align}
			Let us first consider the first term on the right hand side of \Cref{eq6}. By \Cref{eq3} it is equal to
			\begin{align*}
				\int_0^\infty \frac{\e^{\omega_j/2}}{\e^{\omega_j}\sigma+t}\sigma^{1/2}\tilde{L}_j f\sigma^{1/2}(t+\sigma)^{-1}dt&=\int_0^\infty \frac{1}{\e^{\omega_j}\sigma+\e^{\omega_j/2}u}\Gamma_\sigma(\tilde{L}_jf)\frac{\e^{\omega_j}}{\e^{\omega_j/2}u+\sigma}du\\
				&=\int_0^\infty \frac{\e^{-\omega_j/2}}{\e^{\omega_j/2}\sigma+u}\Gamma_\sigma(\tilde{L}_jf)\frac{\e^{-\omega_j/2}}{u+\e^{-\omega_j/2}\sigma}\e^{\omega_j}du\\
				&=[\sigma]_{-\omega_j}^{-1}\circ\Gamma_\sigma(\tilde{L}_j f),
			\end{align*}	
			where we made the change of variable $\e^{\omega_j/2}u=t$ on the first line, and used \Cref{lemma} in the last line. Similarly, using \reff{eq4}, the second term on the right hand side of \Cref{eq6} is equal to
			\begin{align*}
				\int_0^\infty \frac{\e^{-\omega_j/2}}{t+\sigma}\Gamma_\sigma(f\tilde{L}_j)\frac{1}{\e^{-\omega_j}\sigma+t}dt&=\int_0^\infty \frac{1}{u+\e^{\omega_j/2}\sigma}\Gamma_\sigma(f \tilde{L}_j)\frac{1}{\e^{-\omega_j/2}\sigma+u}du\\
				&=[\sigma]_{-\omega_j}^{-1}(\Gamma_\sigma(f\tilde{L}_j)),
				\end{align*}
where we made the change of variable $t=\e^{-\omega_j/2}u$. Hence \Cref{eq5} follows.
			\qed
			\end{proof}
		\noindent	We are now ready to prove \Cref{talagrandpoincare}. Start by the following:		
\begin{align*}
\langle f,\Xi_\sigma(f)\rangle_{1/2,\sigma}=	\tr(\Gamma_{\sigma}(f)\Xi_\sigma(f))&=\frac{1}{\eps}\tr(\sigma^{1/2}\Xi_\sigma(f)\sigma^{1/2}(f^{\eps}-\mathbb{I}))\\
&=\frac{1}{\eps}\tr(\Xi_\sigma(f)(\sigma^{1/2}f^\eps\sigma^{1/2}-\sigma)).
\end{align*}
 For any $\delta>0$, there exists a smooth path $(\gamma^\eps(s))_{s\in[0,1]}$, with associated vector field $(\mathbf{V}^\eps(s))_{s\in[0,1]}$ (cf. \Cref{continuity2}), interpolating between $\rho^\eps:=\Gamma_\sigma(f^\eps)$ and $\sigma$, and such that
 \begin{align}\label{eq16}
 	\left(\int_0^1 \|\dot{\gamma}^{\eps}(s)\|^2_{g_{\LL,\gamma^{\eps}(s)}}\right)^{1/2}\le W_{2,\LL}(\rho^{\eps},\sigma)+\delta
 	\end{align}
 This implies that 
 \begin{align}
 \langle f,\Xi_\sigma(f)\rangle_{1/2,\sigma}&=-\frac{1}{\eps}\tr\left( \Xi_\sigma(f)\int_0^1 \frac{d}{ds}{\gamma}^\eps(s)ds\right)
 	=\frac{1}{\eps}\tr\left(\Xi_\sigma( f)\int_0^1 \operatorname{div}([\gamma^\eps(s)]_{\vec{\omega}} \mathbf{V}^\eps(s))ds\right)\nonumber\\
 	&=-\frac{1}{\eps}\int_0^1\sum_j c_j\langle \partial_j \Xi_\sigma(f),[\gamma^{\eps}(s)]_{\omega_j}(V^\eps(s))_j\rangle_{_{HS}} ds\nonumber\\
 	&\le  \frac{1}{\eps}\left(\sum_j c_j\int_0^1\langle\partial_j\Xi_\sigma(f), [\gamma^\eps(s)]_{\omega_j}\circ \partial_j\Xi_\sigma(f)\rangle_{_{HS}}  ds\right)^{1/2}\left(\int_0^1 \|\mathbf{V}^\eps(s)\|^2_{\LL,\gamma^\eps(s)}ds\right)^{1/2}\nonumber\\
 	&\le \left(\sum_j c_j\int_0^1\langle\partial_j\Xi_\sigma(f), [\gamma^\eps(s)]_{\omega_j}\circ \partial_j\Xi_\sigma(f)\rangle_{_{HS}}  ds\right)^{1/2}\frac{W_{2,\LL}(\rho^\eps,\sigma)+\delta}{\eps} \nonumber\\
 	&\le \left(\sum_j c_j\int_0^1\langle\partial_j\Xi_\sigma(f), [\gamma^\eps(s)]_{\omega_j}\circ \partial_j\Xi_\sigma(f)\rangle_{_{HS}}  ds\right)^{1/2}\frac{\sqrt{2c_2D(\rho^{\eps}\|\sigma)}+\delta}{\eps},\label{eq2}
 	\end{align}
where the first inequality comes from a use of the Cauchy-Schwarz inequality with respect to the inner product $\sum_{j\in \cJ}c_j\langle .~,\int_0^1[\gamma^\eps(s)]_{\omega_j}ds~.\rangle_{_{HS}} $, the second from \Cref{eq16} as well as \Cref{gammav}, and the last one from \Cref{t2}. As $\eps\to 0$, the term in brackets in \reff{eq2} converges to
 	\begin{align*}
 		\sum_j c_j \langle \partial_j\Xi_\sigma(f),[\sigma]_{\omega_j}\circ\partial_j\Xi_\sigma(f)\rangle_{_{HS}} &=\sum_j c_j \langle \partial_j\Xi_\sigma(f),[\sigma]_{\omega_j}\circ[\sigma]_{-\omega_j}^{-1}\circ\Gamma_\sigma\circ\partial_jf\rangle_{_{HS}} \\
 		&\le \sup_{j\in\cJ}\|[\sigma]_{\omega_j}\circ [\sigma]^{-1}_{-\omega_j}\|_{\infty\to \infty}~\mathcal{E}_{1/2,2}(f,\Xi_\sigma(f)),
 		\end{align*}
 		where we used \Cref{lemma8} as well as \Cref{2dirichletform}. Denote
 	$\kappa_\LL:=\sup_{j\in\mathcal{J}}\|[\sigma]_{\omega_j}\circ [\sigma]_{-\omega_j}^{-1}\|_{\infty\to \infty}$. As $\delta>0$ was chosen arbitrarily, we can now take the limit $\delta\to 0$. Moreover, following the approach of the proof of Theorem 16 of \cite{[KT13]}, one can prove that
	\begin{align*}
		D(\rho^\eps\|\sigma)/\eps^2\to \frac{1}{2}(\tr(\Gamma_\sigma(f)~ \Xi_\sigma(f)))
\end{align*}
 Substituting into \reff{eq2}, we get
	\begin{align*}
	\frac{1}{c_2~\kappa_\LL}\langle f,\Xi_\sigma(f)\rangle_{1/2,\sigma}	\le \mathcal{E}_{1/2,2}(f,\Xi_\sigma(f)).
		\end{align*}
	This is exactly the form that was derived at the end of the proof of Theorem 16 of \cite{[KT13]} which led to the Poincar\'{e} inequality.
\qed
\end{proof}
\bigskip
\begin{remark}
	In the classical, commutative case, \Cref{talagrandpoincare} reduces to Proposition 7.6 of \cite{[EM12]}. Indeed, in this case, one can easily verify that for any $j\in\cJ$,  $[\sigma]_{\omega_j}(f)=\frac{1}{2}\sinh(\omega_j/2)\sigma f$ and $[\sigma]_{-\omega_j}^{-1}(f)=2f/(\sigma\sinh(\omega_j/2))$. Therefore $[\sigma]_{\omega_j}\circ [\sigma]_{-\omega_j}^{-1}(f)=f$, and the result follows.
\end{remark}	
\section{Quantum concentration inequalities}\label{qconc}
The following theorem is a generalization of the classical results of \cite{[GM83]} (see also the review \cite{[M09]}). It states that the Poincar\'{e} inequality implies exponential concentration.
\begin{theorem}\label{poincare}	Let $\sigma$ be a non-degenerate, full-rank state, and $(\Lambda_t)_{t\ge 0}$ be a primitive quantum Markov semigroup on $\cB(\cH)$ whose generator $\LL$ is self adjoint with respect to $\langle .,. \rangle_{1,\sigma}$. If $(\Lambda_t)_{t\ge0}$ satisfies (\ref{p}), for a given function $\varphi:(0,\infty)\to (0,\infty)$, then for any self-adjoint operator $f$,
	\begin{align*}
		\tr(\sigma\mathbf{1}_{[r,\infty)}(f-\tr(\sigma f)))\le 3 \e^{-r  \sqrt{\lambda}/(\|f\|_{\operatorname{Lip}}C_{f,\lambda})}.
	\end{align*}
	where $\|.\|_{\operatorname{Lip}}$ is defined in \Cref{lipnorm}, and $C_{f,\lambda}\equiv \frac{\e^{2\sqrt{\lambda} \|f\|_{\infty}/\|f\|_{\operatorname{Lip}}}-1}{\sqrt{2\lambda}\|f\|_{\infty}/\|f\|_{\operatorname{Lip}}}$.
\end{theorem}
\begin{proof}
	Assume without loss of generality that $\tr(\sigma f)=0$. For $\theta\ge 0$, and $f\ne 0$ self-adjoint, let
	\begin{align*}
		M_f(\theta):= \tr\left(\sigma \e^{\theta f}\right).
	\end{align*}
	By \Cref{onpoincare}, one can reduce to the case of $\varphi=\varphi_1$ without loss of generality. In this case, the Poincar\'{e} inequality \ref{p} applied to ${\e^{\theta f/2}}-\tr\sigma\e^{\theta f /2}$ can be written as:
	\begin{align}\label{poinc}
		{\lambda}~(M_f(\theta)-M_f(\theta/2)^2)\le \cE_{1,2}\left(\e^{\theta f/2},\e^{\theta f/2}\right).
	\end{align}
	However, as by assumption $(\Lambda_t)_{t\ge 0}$ is self-adjoint with respect to $\langle.,.\rangle_{1,\sigma}$, one can rewrite $\mathcal{E}_{1,2}$ as follows (cf. \Cref{2dirichletform}):
	\begin{align}
		\cE_{1,2}\left(\e^{\theta f/2},\e^{\theta f/2}\right)&=\sum_{j\in\cJ}c_j\e^{-\omega_j/2}\tr \left[\sigma\left( \partial_j \e^{\theta f/2}\right)^*\partial_j\e^{\theta f/2}\right]\nonumber\\
		&=\frac{\theta^2}{4}\sum_{j\in\cJ}c_j\e^{-\omega_j/2}\iint_{[0,1]^2} \tr \left(\sigma  \e^{\frac{(1-s) \theta f}{2}}(\partial_jf)^* \e^{\frac{s\theta f}{2}}\e^{\frac{u\theta f}{2}}\partial_j f\e^{\frac{(1-u)\theta f}{2}}\right)duds.\label{eq1}
	\end{align}
	where we used \Cref{chainrule} in the second line. Moreover, for each $u,s\in[0,1]$, the trace in \Cref{eq1} is equal, by cyclicity,  to
	\begin{align*}
		\tr	&\left[\left(\e^{\theta f/2}\sigma \e^{\theta f/2}\right)^{*}\left(\e^{ \frac{-s\theta f}{2}}(\partial_jf)^* \e^{\frac{(s+u)\theta f}{2}}\partial_j f\e^{\frac{-u\theta f}{2}}\right)\right]\\
		&	\le \tr\left(\e^{\theta f/2}\sigma \e^{\theta f/2}\right)\e^{\theta (s+u)\|f\|_\infty}\|(\partial_j f)^*\|_\infty\|\partial_j f\|_\infty\\
		&=M_f(\theta)\e^{\theta(s+u) \|f\|_\infty} \|\partial_j f\|_\infty^2,
	\end{align*}
	where we used H\"{o}lder's inequality as well as the submultiplicativity of the operator norm in the second line. Substituting into \reff{eq1}, we thus get:
	\begin{align*}
		\mathcal{E}_{1,2}\left(\e^{\theta f/2},\e^{\theta f/2}\right)&\le \frac{M_f(\theta)}{4\|f\|_\infty^2}\sum_{j\in\cJ} c_j\e^{-\omega_j/2}\|\partial_j f\|_\infty^2\left(\e^{\theta \|f\|_\infty}-1\right)^2\\
		&\le \|f\|_{\operatorname{Lip}}^2~M_f(\theta)\frac{\left(\e^{\theta \|f\|_\infty}-1\right)^2}{4\|f\|_\infty^2},
	\end{align*}
	However, for any $0\le \theta<2\sqrt{\lambda}/\|f\|_{\operatorname{Lip}}$,
	\begin{align*}
		\frac{\e^{\theta \|f\|_{\infty}}-1}{\theta\|f\|_{\infty}}\le \frac{\e^{2\sqrt{\lambda} \|f\|_{\infty}/\|f\|_{\operatorname{Lip}}}-1}{\sqrt{2\lambda}\|f\|_{\infty}/\|f\|_{\operatorname{Lip}}}\equiv  C_{f,\lambda}>1.
	\end{align*}
	Hence, substituting into \Cref{poinc}:
	\begin{align*}
		{\lambda} ~(M_f(\theta)-M_f(\theta/2)^2)\le {\theta^2 \|f\|^2_{\operatorname{Lip}} ~C_{f,\lambda}^2}~ M_f(\theta)/4.
	\end{align*}
	This last inequality implies that
	\begin{align*}
		M_f(\theta)\le\frac{1}{1-\theta^2 \|f\|_{\operatorname{Lip}}^2~C_{f,\lambda}^2/{(4\lambda)}}M_f(\theta/2)^2,
	\end{align*}
	for every $\theta< 2\sqrt{\lambda}/(C_{f,\lambda} \|f\|_{\operatorname{Lip}})$. A simple iteration procedure yields
	\begin{align*}
		M_f(\theta)\le \prod_{k=0}^{n-1} \left( \frac{1}{1-\theta^2 \|f\|_{\operatorname{Lip}}^2~C_{f,\lambda}^2/(4^{k+1}{\lambda})}\right)^{2^k}M_f(\theta/2^n)^{2^n}.
	\end{align*}
	Note that $M_f(\theta)=1+\theta\tr(\sigma f)+\mathcal{O}(\theta^2)$, and we have assumed that $\tr(\sigma f)=0$. Thus letting $n\to \infty$:
	\begin{align*}
		M_f(\theta)\le \prod_{k=0}^{\infty} \left( \frac{1}{1-\theta^2 \|f\|_{\operatorname{Lip}}^2~C_{f,\lambda}^2/(4^{k+1}{\lambda})}\right)^{2^k}.
	\end{align*}
	Set $\theta=\sqrt{\lambda}/(\|f\|_{\operatorname{Lip}}C_{f,\lambda})$, then the right hand side is a universal constant contained between $\e$ and $3$. So we proved that
	\begin{align*}
		M_f\left(\sqrt{\lambda}/(\|f\|_{\operatorname{Lip}}C_{f,\lambda})\right)\le 3.
	\end{align*}
	Now by functional calculus, for any $r\in\RR$ and $\theta>0$:
	\begin{align*}
		\mathbf{1}_{[r,\infty)}(f)&=	\mathbf{1}_{[\exp(\theta r),\infty)}(\exp(\theta f)) \le \e^{-\theta r}\e^{\theta f}.
	\end{align*}	
	This leads to the following Markov-type inequality:
	\begin{align}\label{markov}
		\tr(\sigma  \mathbf{1}_{[r,\infty)}(f))\le \e^{-r\theta}\tr(\sigma \exp(\theta f))=\e^{-r\theta}M_f(\theta).
	\end{align}
	Therefore
	\begin{align*}
		\tr(\sigma\mathbf{1}_{[r,\infty)}(f))\le 3 \e^{-r \sqrt{\lambda}/(\|f\|_{\operatorname{Lip}}C_{f,\lambda})}.
	\end{align*}
	\qed
\end{proof}
\noindent The first proof that the classical transportation cost inequality of order $1$ implies Gaussian concentration is due to Marton \cite{[M96]}. The following theorem is a quantum generalization of Bobkov-G\"{o}tze's proof \cite{[BG99]} which relies on the variational representations of the $1$ Wasserstein distance (see also Theorem 36 of \cite{[RS14]} or Proposition 7.7 of \cite{[EM12]}):
\begin{theorem}\label{talconc}		
	Let $\sigma$ be a full-rank state, and $(\Lambda_t)_{t\ge 0}$ be a primitive quantum Markov semigroup on $\cB(\cH)$ whose generator $\LL$ is self adjoint with respect to $\langle .,. \rangle_{1,\sigma}$. If $(\Lambda_t)_{t\ge 0}$ satisfies \reff{t1}, the following Gaussian concentration inequality holds: for any self-adjoint operator $f$,
	\begin{align*}
		\tr(\sigma\mathbf{1}_{[r,\infty)}(f-\tr(\sigma f)))\le \exp\left({{-\frac{r^2}{8\max (\|(\Delta_\sigma^{-1/2}f)_R\|_{\operatorname{Lip}}^2,\|(\Delta_\sigma^{-1/2}f)_I\|^2_{\operatorname{Lip}})c_1}}}\right),
	\end{align*}
	where $(\Delta_{\sigma}^{-1/2}f)_R$, resp. $(\Delta_{\sigma}^{-1/2}f)_I$ are the real and imaginary parts of $\Delta_{\sigma}^{-1/2}f$.
\end{theorem}	
\begin{proof} Here we follow the lines of the proof of Theorem 36 of \cite{[RS14]}. Let $g:=g_R+i g_I\in\cB(\cH)$ be the decomposition of an operator $g$ into its real and imaginary parts, where $g_R,g_I\in\cB_{sa}(\cH)$. Assume moreover that $\tr(\sigma g)=0$, and $\|g_R\|_{\operatorname{Lip}}, \|g_I\|_{\operatorname{Lip}}  \le  1$. From (\ref{t1}), we know that for any $\rho\in\cD_+(\cH)$,  
	\begin{align*}
	|	\tr(\rho g)|\le | \tr (\rho~ g_R) |+|\tr (\rho~ g_I)|   \le 2 ~W_{1,\LL}(\rho,\sigma)\le 2 \sqrt{2c_1D(\rho\|\sigma)}.
	\end{align*} 
	Next, from the fact that
	\begin{align*}
		\inf_{\theta>0}\left(\frac{a}{\theta}+\frac{b\theta}{2}\right)=\sqrt{2ab}
	\end{align*}
	for any $a,b\ge 0$, we see that any such $g$ must satisfy
	\begin{align*}
	|	\tr(\rho g)|\le \frac{2}{\theta}D(\rho\|\sigma)+c_1\theta, ~~~\forall \theta>0.
	\end{align*}
	Rearranging, we obtain
	\begin{align}\label{eq7}
		\theta	|\tr(\rho g)|-c_1\theta^2\le 2D(\rho\|\sigma)\le 2\widehat{D}(\rho\|\sigma),~~~\forall \theta>0,
	\end{align}
	where we have used \reff{DhD}, and $\widehat{D}(\rho\|\sigma)$ is the maximal divergence defined though \Cref{maxi}. Define $\rho:=\sigma^{1/2} \e^{\theta f}\sigma^{1/2}/(\tr(\sigma \e^{\theta f}))$, where $f$ is a self-adjoint operator to be specified later. Hence \Cref{eq7} becomes
	\begin{align*}
		\theta|\tr(\rho g)|-c_1\theta^2\le 2 \theta \frac{\tr(\sigma \e^{\theta f}f)}{\tr(\sigma\e^{\theta f})}-2\log(\tr(\sigma\e^{\theta f})).
	\end{align*}
	Now for $f=\frac{1}{2}\Delta_{\sigma}^{1/2}(g)$, the last expression further simplifies into
	\begin{align*}
		-\frac{c_1 \theta^2}{2}\le-\log(\tr(\sigma\e^{\theta f}))~~~\Rightarrow ~~~M_f(\theta)\equiv\tr(\sigma \e^{\theta f})\le\e^{c_1\theta^2/2}.
	\end{align*}
	Note that the quantity $M_f(\theta)$ was also used in the proof of \Cref{poincare}. As in \reff{markov}:
	\begin{align}\label{markov2}
		\tr(\sigma  \mathbf{1}_{[r,\infty)}(f))\le \e^{-r\theta}M_f(\theta)\le \e^{-r\theta}\e^{c_1\theta^2/2}.
	\end{align}
	Optimizing over all $\theta>0$,
	\begin{align*}
		\tr(\sigma\mathbf{1}_{[r,\infty)}(f))\le \e^{-\frac{r^2}{2c_1}}.
	\end{align*}
	In order to achieve this bound we assumed that $g=2 \Delta_\sigma^{-1/2}(f)\in\cB(\cH)$ is such that $\|g_R\|_{\operatorname{Lip}},~\|g_I\|_{\operatorname{Lip}}\le 1$ and $\tr(\sigma g)=0$. This implies that $\tr(\sigma f)= \frac{1}{2}\tr(\sigma \Delta^{1/2}_\sigma(g))=\frac{1}{2}\tr(\sigma g)=0$. 
	Therefore, up to a rescaling we proved that for any self-adjoint operator $f$,
	\begin{align*}
		\tr(\sigma\mathbf{1}_{[r,\infty)}(f-\tr(\sigma f)))\le \exp\left({-\frac{r^2}{8\max(\|(\Delta_\sigma^{-1/2}f)_R\|^2_{\operatorname{Lip}},\|(\Delta_\sigma^{-1/2}f)_I\|^2_{\operatorname{Lip}})c_1}}\right).
	\end{align*}
	\qed
\end{proof}
\begin{remark}
	In the commutative case, $\Delta_\sigma^{-1/2}f=f$, so that for $$f=\sum_{i=1}^d \varphi(i)|i\rangle\langle i|,~~~~~~~\|(\Delta_\sigma^{-1/2}f)_{R,I}\|_{\operatorname{Lip}}=\|f\|_{\operatorname{Lip}}\le\|\varphi\|_{\operatorname{lip,g}}.$$ We therefore recover the classical Gaussian concentration inequality, up to a factor $\frac{1}{4}$.
\end{remark}	
\begin{remark}\label{rem10}
	The same kind of concentration result follows if we replace \reff{t1} by a transportation cost inequality associated to any other Wasserstein distance defined in \Cref{WALT}:
	\begin{align*}
	W_{1,\LL,\star}(\rho,\sigma)\le \sqrt{2c_1 D(\rho\|\sigma)}.
	\end{align*}
	\end{remark}
\subsection*{Example: Generalized quantum depolarizing semigroups.}

In \cite{[MFW16]}, the authors computed the log-Sobolev constant $\alpha_1$ for the so-called generalized quantum depolarizing semigroups. Let $\cH$ be a Hilbert space of dimension $d$. Given a state $\sigma\in\cD_+(\cH)$, the generator of the quantum depolarizing semigroup of invariant state $\sigma$, denoted by $\LL_\sigma:\cB(\cH)\to \cB(\cH)$, is defined as
\begin{align}\label{depol}
\LL_\sigma(f):=\tr(\sigma f)\mathbb{I}-f, ~~~~~~~~~f\in\cB(\cH).
\end{align}
One can verify that, in the Schr\"{o}dinger picture, the associated semigroup $\Lambda_{*t}^{\sigma}$ has the following expression: for any $\rho\in\cD_+(\cH)$,
\begin{align*}
\Lambda_{*t}^{\sigma}(\rho)=(1-\e^{-t})~\sigma+\e^{-t}\rho,
\end{align*}
hence implementing a generalized depolarizing channel (i.e.~$\mathbb{I}/d$ replaced by $\sigma$) of error probability $(1-\e^{-t})$ at each time $t\ge 0$.  The semigroup $(\Lambda^{\sigma}_t)_{t\ge 0}$ is primitive and its generator is self-adjoint with respect to $\langle .,.\rangle_{1,\sigma}$. Indeed
\begin{align*}
\langle f,\LL_\sigma(g)\rangle_{1,\sigma}=\tr (\sigma f^*\LL_\sigma(g))&=\tr(\sigma f^*) \tr(\sigma g)  -\tr(\sigma f^*g)\\
&=\tr(\sigma f)^*\tr(\sigma g)-\tr(\sigma f^*g)\\
&=\tr(\sigma(\tr(\sigma f)\mathbb{I}-f)^*g)\\
&=\langle \LL_\sigma(f),g\rangle_{1,\sigma}.
\end{align*}
In \cite{[MFW16]}, the authors proved that the largest log-Sobolev constant $\alpha_1(\sigma)$ of the semigroup $(\Lambda^{\sigma}_t)_{t\ge 0}$ is equal to 
\begin{align}\label{alphadepo}
\alpha_1(\sigma)= \min_{x\in[0,1]}\frac{1}{2}(1+ q(x,\sigma_{\min})),
\end{align}
where $\sigma_{\min}$ stands for the smallest eigenvalue of $\sigma$, and for $x\in [0,1]$, $y\in(0,1)$,
\begin{align*}
q(x,y)= \left\{\begin{aligned}
&\frac{D_2(y\|x)}{D_2(x\|y)}~~~~~x\ne y,\\
& 1~~~~~~~~~~~~~~~x=y,
\end{aligned}\right.
\end{align*}
where $D_2(x\|y):=x\log(x/y)-(1-x)\log ((1-x)/(1-y))$ is the binary relative entropy. We now prove that $\LL_\sigma$ can be expressed in the form of \Cref{LLDBC}, where here the sum is over a pair of indices $(i,j)$. Given the following eigenvalue decomposition $\sigma:=\sum_{i=1}^d \sigma_i|i\rangle\langle i|$, define the operators $L_{ij}:=\sqrt{\sigma_i}|i\rangle \langle j|$. Hence for any $f\in\cB(\cH)$, $\tr(\sigma f)\mathbb{I}= \sum_{i,j=1}^d L_{ij}^* fL_{ij}$, so that
\begin{align}\label{depo}
\LL_\sigma(f)=-\frac{1}{2}\sum_{i,j=1}^d  L_{ij}^* L_{ij}f-2L_{ij}^*fL_{ij}+f L_{ij}^*L_{ij}.
\end{align}
Moreover, $\Delta_{\sigma}(L_{ij})=\sigma_i/\sigma_j L_{ij}$, so that $\omega_{ij}\equiv \log \sigma_j-\log\sigma_i$. Therefore, for any $f\in\cB_{sa}(\cH)$, \Cref{depo} can be rewritten in the form of \Cref{LLDBC} by taking $\tilde{L}_{ij}=\sqrt{d}|i\rangle\langle j|$, and $c_{ij}=\frac{1}{2d}\sqrt{\sigma_i\sigma_j}$. Moreover, a straightforward extension of the proof of \Cref{pinsker} provides the following bound: for any self-adjoint operator $f:=\sum_{x}\varphi(x)|e_x\rangle\langle e_x|$, and any state $\gamma\in\cD_+(\cH)$,
\begin{align*}
\sum_{i,j=1}^d& c_{ij}(\e^{-\omega_{ij}/2}  \tr(\gamma \partial_{ij}f^*\partial_{ij} f )+ \e^{\omega_{ij}/2} \tr(\gamma \partial_{ij}f \partial_{ij}f^*  ))\\
&=\sum_{xy} \varphi(x)\varphi(y)\left\{ \delta_{xy}(\langle e_x|\gamma|e_x\rangle + \langle e_x|\sigma|e_x\rangle)-\langle e_x|\sigma|e_x\rangle \langle e_y|\gamma|e_y\rangle-\langle e_y|\sigma |e_y\rangle \langle e_x|\gamma|e_x\rangle\right\}\\
&=-\frac{1}{2}\sum_{xy}(\varphi(x)-\varphi(y))^2\{   \delta_{xy}(\langle e_x|\gamma|e_x\rangle + \langle e_x|\sigma|e_x\rangle)-\langle e_x|\sigma|e_x\rangle \langle e_y|\gamma|e_y\rangle-\langle e_y|\sigma |e_y\rangle \langle e_x|\gamma|e_x\rangle  \}\\
&\le 2 \sup_{x\ne y}( \varphi(x)-\varphi(y))^2 .
\end{align*}
Therefore, replacing $\gamma$ by $\gamma(s)$ as in the proof of \Cref{W1W2}, and replacing the last line of the proof by the bound we have just proved:
\begin{align}\label{eq99}
|\tr f(\rho-\sigma)|\le ~\sqrt{2}\sup_{x\ne y}|\varphi(x)-\varphi(y)|~ W_{2,\LL_\sigma}(\rho,\sigma)\equiv \sqrt{2} \|\varphi\|_{\operatorname{lip,H}} W_{2,\LL_\sigma}(\rho,\sigma).
\end{align}
Taking the supremum  over self-adjoint operators $f:=\sum_{x}\varphi(x)|e_x\rangle \langle e_x|$ such that $\|\varphi\|_{\operatorname{lip,H}}\le 1$, we define the following \textit{modified Wasserstein distance of order 1}:
\begin{align*}
W_{1,\operatorname{cl}}(\rho,\sigma):=\sup_{\substack{f=\sum_{x}\varphi(x)|e_x\rangle\langle e_x|\in\cB_{sa}(\cH)\\\|\varphi\|_{\operatorname{lip,H}}\le 1}}|\tr f(\rho-\sigma)|,
\end{align*}
where the subscripts $\operatorname{cl}$ denotes the fact that the optimum is taken over states $|e_x\rangle$ and functions $\varphi$ with classical Lipschitz norm bounded by $1$. We have hence proved the following:
\begin{theorem}\label{propdepol} Let $\sigma$ be a full-rank state, and denote by $\LL_\sigma$ the generator of the generalized depolarizing semigroup with invariant state $\sigma$. Hence, for any $\rho\in\cD_+(\cH)$, the following holds:
	\begin{align}\label{eq30}
	\|\rho-\sigma\|_1\le W_{1,\operatorname{cl}}(\rho,\sigma)\le\sqrt{2}  W_{2,\LL_\sigma}(\rho,\sigma).
	\end{align}
	Moreover, for any self-adjoint operator $f:=\sum_x \varphi(x)|e_x\rangle\langle e_x|$,
	\begin{align*}
\tr(\sigma\mathbf{1}_{[r,\infty)}(f-\tr(\sigma f)))\le \exp\left(-  \frac{r^2 \alpha_1(\sigma)}{16\max(\|\varphi_R\|_{\operatorname{lip,H}}^2,\|\varphi_I\|^2_{\operatorname{lip,H}})}\right),
\end{align*}	
where $\alpha_1(\sigma)$ is given in \Cref{alphadepo}, and $\varphi_R$, resp, $\varphi_I$, is such that $$(\Delta_\sigma^{-1/2}f)_R=\sum_x \varphi_R(x)|e^R_x\rangle\langle e^Rx|,~~~~~(\Delta_\sigma^{-1/2}f)_I=\sum_x \varphi_I(x)|e^I_x\rangle\langle e^Ix|,$$ where $(\Delta_\sigma^{-1/2}f)_R$, $(\Delta_\sigma^{-1/2}f)_I$ are the real and imaginary parts of $\Delta_\sigma^{-1/2}f$.
\end{theorem}	
\begin{proof}
The result follows from \Cref{eq99} and \Cref{logsobtalagrand} together with a straightforward adaptation of \Cref{talconc}.
\qed
\end{proof}
\bigskip
\begin{remark}
	\Cref{propdepol} is to be compared with \Cref{talconc}. In a nutshell, had we used \Cref{talconc} directly to get our concentration bound for the invariant state of a generalized depolarizing semigroup, we would have ended up with with a dependence on the dimension $d$ of the Hilbert space, due to the passage from \reff{t2} to \reff{t1} (cf. \Cref{t2t1}). Here, we showed that a finer analysis of this semigroup leads to the removal of the dimensional factor, as it can be seen by comparing \reff{eq30} with \Cref{W1W2}.
\end{remark}	

\subsection*{Concentration of measure for product states}

In \cite{[TPK14]}, the authors proved that the so-called $2$-log-Sobolev constant $\alpha_2$ satisfies the following property: let $\cH$ be a finite dimensional Hilbert space of dimension $d$ and, for each $k=1,...,n$, let $\sigma_k$ be a full-rank state and $\LL_k:\cB(\cH)\to \cB(\cH)$ be the generator of a primitive semigroup, where each $\LL_k$ is primitive and self-adjoint with respect to $\langle.,.\rangle_{1/2,\sigma_k}$. Then the generator 
\begin{align*}
	\LL:=\sum_{k=1}^n \id^{\otimes (k-1)}\otimes \LL_k\otimes \id^{\otimes (n-k)}
\end{align*}
is self-adjoint with respect to $\langle.,.\rangle_{1/2,\sigma^{(n)}}$, where $\sigma^{(n)}:=\sigma_1\otimes...\otimes \sigma_n$. Moreover, denoting by $\lambda_k$ the spectral gap of each $\LL_k$, the $2$-log Sobolev constant $\alpha_2$ of $\LL$ is bounded as follows:
\begin{align}\label{lambda}
	\frac{\lambda}{\log(d^4 s)+11}\le \alpha_2\le \lambda,
\end{align}
where $\lambda:= \min_k \lambda_k$, ans $s:=\max_k \|\sigma_k^{-1}\|_{\infty} $ (see Theorem 9 of \cite{[TPK14]}). Moreover, it was shown in Proposition 13 of \cite{[KT13]} that the generator of a primitive semigroup satisfies the following inequality:
\begin{align}\label{alpha}
	\alpha_2\le \alpha_1
\end{align}
provided it is strongly $\operatorname{L}_p$ regular (see Definition 9 of \cite{[KT13]}). Strong L$_p$ regularity always holds for classical semigroups, and was also shown to hold true for Davies generators (see Theorem 20 of \cite{[KT13]}). Therefore, by a joint use of \Cref{logsobtalagrand,t2t1,talconc} as well as \reff{lambda} and \reff{alpha}, the following holds true: for any self-adjoint operator $f_n$ on $\cH^{\otimes n}$,
\begin{align}
	&\tr(\sigma^{(n)}\mathbf{1}_{[r,\infty)}(f_n-\tr(\sigma^{(n)} f_n)))\label{eq9}\\
	&~~~~~~~~~~~~~~~~~~~~~~~~~~~\le \exp\left({-\frac{\lambda r^2}{8d(11+\log (d^4 s))\max(\|(\Delta_\sigma^{-1/2}f)_R\|^2_{\operatorname{Lip}},\|(\Delta_\sigma^{-1/2}f)_I\|^2_{\operatorname{Lip}})}}\right)\nonumber
\end{align}
Assume now that $f_n$ has the following form:
\begin{align*}
	f_n:=\frac{1}{n}\sum_{k=1}^n \mathbb{I}^{\otimes (k-1)}\otimes f\otimes \mathbb{I}^{\otimes (n-k)}.
\end{align*}
Assume moreover that for each $k$, $\LL_k$ is of the form given in \Cref{LLDBC}:
\begin{align*}
	\LL_k(f_n):=	\sum_{j\in \mathcal{J}_k}c^{(k)}_j\e^{-\omega^{(k)}_j/2}\left( \tilde{L}^{(k)*}_j[f_n,\tilde{L}^{(k)}_j]+[\tilde{L}^{(k)*}_j,f_n]\tilde{L}^{(k)}_j\right),
\end{align*}	
where for each $k,j$, we assimilate $\tilde{L}_j^{(k)}$ with $ \mathbb{I}^{\otimes (k-1)}\otimes \tilde{L}_j^{(k)}\otimes \mathbb{I}^{\otimes (n-k)}$ by abuse of notation. Hence, the Lipschitz constants in \reff{eq9} reduces to:
\begin{align*}
	\|(\Delta_{\sigma^{(n)}}^{-1/2} f_n)_{R,I}\|_{\operatorname{Lip}}^2&:=\sum_{k=1}^n\sum_{j\in \mathcal{J}_k} c_{j}^{(k)} (\e^{-\omega_j^{(k)}/2}+\e^{\omega_j^{(k)}/2})  \| [\tilde{L}_j^{(k)},(\Delta_{\sigma_k}^{-1/2}f)_{R,I}]/n\|_\infty^2\\
	&=\frac{1}{n^2}\sum_{k=1}^n\sum_{j\in \mathcal{J}_k} c_{j}^{(k)} (\e^{-\omega_j^{(k)}/2}+\e^{\omega_j^{(k)}/2})  \| [\tilde{L}_j^{(k)},(\Delta_{\sigma_k}^{-1/2}f)_{R,I}]\|_\infty^2.
\end{align*}
Assume finally that the generators $\LL_k$ are identical, with associated invariant state $\sigma_k\equiv\sigma$. Then,
\begin{align*}
	\|(\Delta_{\sigma^{(n)}}^{-1/2} f_n)_{R,I}\|_{\operatorname{Lip}}^2= \frac{1}{n} \| (\Delta_{\sigma}^{-1/2}f)_{R,I}\|_{\operatorname{Lip}}^2.
\end{align*}
In this case, \Cref{eq9} reduces to
\begin{align}\label{eq14}
	&\tr(\sigma^{\otimes n}\mathbf{1}_{[r,\infty)}(f_n-\tr(\sigma^{\otimes n} f_n)))\\
	&~~~~~~~~~~~~~~~~~~~~~~~~~~~~~~\le \exp\left({-\frac{\lambda n r^2}{8d(11+\log (d^4 \|\sigma^{-1}\|_\infty))\max(\|(\Delta_\sigma^{-1/2}f)_R\|^2_{\operatorname{Lip}},\|(\Delta_\sigma^{-1/2}f)_I\|^2_{\operatorname{Lip}})}}\right).\nonumber
\end{align}

\section{Non-asymptotic quantum parameter estimation}\label{parameter}

Here, we apply \reff{eq14} to the problem of parameter estimation of quantum states. Assume that $n$ independent physical systems are prepared in the same state $\rho_\theta$, where $\theta$ is an unknown parameter belonging to a set $\Theta$. Here, we assume that $\Theta:=\RR$. In order to estimate $\theta$, an estimator is described by a sequence of positive operator valued measurement (POVM in short) $\vec{M}:=\{M^{(n)}\}_{n\in\NN}$, where, for each $n$, $M^{(n)}:\cB(\RR)\mapsto \cP(\cH^{\otimes n})$ is a POVM on the Hilbert space $\cH^{\otimes n}$ associated to the $n$ systems, where $\cB(\RR)$ stands for the Borel algebra associated to $\RR$. The merit of such a POVM can be quantified in terms of the following error exponent (see \cite{[H02],[N05],[H05]}):
\begin{align*}
	\beta(\vec{M},\theta,\eps,n):=-\frac{1}{n\eps^2} \log \PP_{M^{(n)}}(\hat{\theta}_n\in [\theta-\eps,\theta+\eps]^c),
\end{align*}
where $$\PP_{M^{(n)}}(\hat{\theta}_n\in[\theta-\eps,\theta+\eps]^c):=\tr ( M^{(n)}([\theta-\eps,\theta+\eps]^c)\rho_\theta^{\otimes n})$$ is the probability that the estimated value $\hat{\theta}_n$ is at least $\eps$ away from the true parameter $\theta$. In the asymptotic setting $n\to \infty$, it was shown in Lemma 14 of \cite{[H02]} that, under some technical assumptions, any POVM $\vec{M}$ satisfies
\begin{align}\label{povm}
	\limsup_{\eps\to 0}\limsup_{n\to \infty} \beta(\vec{M},\theta,\eps,n)\le \frac{J_{\theta}}{2},
	\end{align}
where $J_\theta:= \tr (\rho_\theta L_\theta^2)$ is the so-called quantum symmetric logarithmic derivative (SLD for short) Fisher information, with associated self-adjoint logarithmic derivative $L_\theta$ defined by $\frac{d}{d\theta}\rho_\theta=\frac{1}{2}(\rho_\theta L_\theta+L_\theta \rho_\theta)$. For sake of simplicity, we assume that for any $\theta\in\Theta$, $\rho_\theta$ is full-rank, so that $L_\theta$ is well-defined. Moreover, the bound in \Cref{povm} was proved to be saturated for a sequence of projection-valued measurements $\vec{M}_\theta$ associated to the self-adjoint operator
\begin{align}\label{eq15}
	f^{(n)}_\theta:= \frac{1}{n}\sum_{k=1}^n  \mathbb{I}^{\otimes (k-1)}\otimes \left(\frac{L_\theta}{J_\theta}+\theta\mathbb{I}\right)\otimes \mathbb{I}^{\otimes (n-k)},
	\end{align}
where the estimated value $\hat{\theta}_n$ is determined to be the outcome of the measurement $M^{(n)}_\theta$. This implies that the error probability asymptotically behaves as
\begin{align*}
	\PP_{M_\theta^{(n)}}(\hat{\theta}_n\in[\theta-\eps,\theta+\eps]^c)\gtrsim
	 \e^{-\eps^2 n J_\theta/2},~~~~~n\to\infty,~\eps\to 0.
\end{align*}
The family $\vec{M}_{\theta}$ forms a sequence of unbiased estimators, i.e. 
\begin{align}\label{unbiased}
\forall n,\forall \theta\in \RR ~\tr\rho_\theta^{\otimes n}f^{(n)}_\theta=\theta.
\end{align}
The following result provides a finite $n$ and finite $\eps$ upper bound on the error probability in the case when $\rho_\theta$ is prepared by means of a dissipative process.
\begin{proposition}\label{prop100}
 Let $\cH$ be a finite dimensional Hilbert space, where $\dim (\cH)=d$. For $\theta\in\RR$, let $\LL_\theta$ be the generator of a quantum Markov semigroup on $\cH$ which is self-adjoint with respect to $\langle .,.\rangle_{1/2,\rho_\theta}$, satisfies \Cref{LLDBC}, and for which $\alpha_2\le \alpha_1$. Then, for any sequence of unbiased projective measurements $\vec{M}:=\{M^{(n)}\}_{n\in\NN}$ associated to the self-adjoint operators
 \begin{align}\label{fn}
 	f_n:=\frac{1}{n}\sum_{k=1}^n\mathbb{I}^{\otimes (k-1)}\otimes f\otimes \mathbb{I}^{\otimes (n-k)},
 	\end{align}
 	where $f$ is a self-adjoint operator on $\cH$, the probability that the estimated value $\hat{\theta}_n$ lies at least $\eps$ away from the true parameter $\theta$ is given by
 \begin{align}
 \PP_{M^{(n)}}(\hat{\theta}_n\in[\theta-\eps,\theta+\eps]^c)\le 2\exp\left({-\frac{ n\eps^2\lambda_\theta }{8d(11+\log(d^4 \|\rho_\theta^{-1}\|_\infty))\max(\|(\Delta_\sigma^{-1/2}f)_R\|^2_{\operatorname{Lip}},\|(\Delta_\sigma^{-1/2}f)_I\|^2_{\operatorname{Lip}})}}\right),\label{eq222}
 	\end{align}
 where $\lambda_\theta$ is the spectral gap of $\LL_\theta$.
\end{proposition}
\begin{proof}
	\begin{align*}\PP_{M^{(n)}}(\hat{\theta}_n\in[\theta-\eps,\theta+\eps]^c)&=\PP_{M^{(n)}}(\hat{\theta}_n\ge \theta+\eps)+\PP_{M^{(n)}}(\hat{\theta}_n\le\theta-\eps)\\
	&= \tr(\rho_\theta^{\otimes n} \mathbf{1}_{[\eps,\infty)}(f^{(n)}-\tr(f^{(n)}\rho_\theta^{\otimes n})))\\
	&+\tr(\rho_\theta^{\otimes n} \mathbf{1}_{(-\infty,-\eps]}(f^{(n)}-\tr(f^{(n)}\rho_\theta^{\otimes n}))),
	\end{align*}
	where we used that $\tr(\rho_\theta^{\otimes n}f_n)=\theta$ for all $n$. We conclude from \Cref{eq14}.
	\qed
	\end{proof}
\noindent
It is a well known fact that the spectral gap of the generalized depolarizing semigroup \reff{depol} is equal to $1$ (see Lemma 25 of \cite{[KT13]}). Therefore, a straightforward use of \Cref{propdepol} leads to:\footnote{This extension of our results was pointed out by  Daniel Stilck Franca} 
\begin{proposition}\label{cor100}
	Let $\cH$ be a finite dimensional Hilbert space, where $\dim (\cH)=d$. For any $\theta\in\RR$, let $\rho_\theta\in\cD_+(\cH)$. Then, for any sequence of unbiased projective estimators $\vec{M}:=\{M^{(n)}\}_{n\in\NN}$ associated to the self-adjoint operators defined in \Cref{fn}, the probability that the associated estimated value $\hat{\theta}_n$ lies at least $\eps$ away from the true parameter $\theta$ is given by
	\begin{align}
\PP_{M^{(n)}}(\hat{\theta}_n\in[\theta-\eps,\theta+\eps]^c)\le 2\exp\left({-\frac{n\eps^2 }{16(11+\log(d^4 \|\rho_\theta^{-1}\|_\infty))        \max(\|\varphi_R\|_{\operatorname{lip,H}},^2\|\varphi_I\|^2_{\operatorname{lip,H}})}}\right),\label{eq102}
	\end{align}
	where $\varphi_R$ and $\varphi_I$ are defined in \Cref{propdepol}.
\end{proposition}	
\begin{proof}
 A straightforward generalization of the computation leading to \Cref{eq99}, replacing the semigroup $(\Lambda_t^\sigma)_{t\ge 0}$ by $n$ uses of itself $(\Lambda^{\sigma^{\otimes n}})_{t\ge 0}:=((\Lambda_t^\sigma)^{\otimes n})_{t\ge 0}$, with associated generator $\LL_{\sigma^{\otimes n}}=\sum_{k=1}^n \id^{\otimes k-1}\otimes \LL_\sigma\otimes \id^{n-k}$, we end up with, for any full-rank state $\rho_n\in\cD_+(\cH^{\otimes n})$,
		\begin{align}\label{eq100}
		|\tr f_n (\rho_n-\sigma^{\otimes n})|\le \sqrt{2}\|\varphi\|_{\operatorname{lip,H}}W_{2,\LL_{\sigma^{\otimes n}}}(\rho_n,\sigma^{\otimes n}),
	\end{align}
	where $f:=\sum_{x}\varphi(x)|e_x\rangle\langle e_x|$. Using \Cref{logsobtalagrand}, \Cref{lambda}, as well as the fact that the generalized depolarizing semigroup is L$_p$ regular (see \cite{[KT13]}) and has spectral gap $\lambda(\sigma)$ equal to $1$, a straightforward generalization of the proof of \Cref{talconc}, where the use of \reff{t1} is replaced by \Cref{eq100}, leads to the result.	 
\qed
\end{proof}	
\begin{remark}
	The above corollary not only provides a bound for any family of states $(\rho_\theta)_{\theta\in\RR}$, but also gives a better dependence of the bound than the one derived in \Cref{prop100} by removing the factor $d$ in the exponent on the right hand side of \reff{eq222}.
\end{remark}	

\section{Summary and open questions}
In this paper we derived concentration inequalities for invariant states of a class of quantum Markov semigroups. More precisely, we define quantum versions \reff{t1},\reff{t2} of the classical transportation cost inequalities TC$_1$ and TC$_2$ (cf.~\reff{tpp}). These inequalities involve two quantum generalizations of the classical Wasserstein distances $W_1$ and $W_2$ (cf.~\reff{2.44}), the latter defined by Carlen and Maas \cite{[CM16]}, and the former being defined in this paper. We then proved that these inequalities are related to quantum functional inequalities, namely the modified log Sobolev inequality and Poincar\'{e} inequality, as well as concentration of quantum states, analogously to their classical counterparts (see \Cref{fig2}). We compared our quantum transportation cost inequalities with their classical counterparts, and showed their relation to both classical and quantum Pinsker's inequalities. We studied the example of the generalized depolarizing semigroup which provides a Gaussian concentration for any full-rank quantum state. Finally, we applied our concentration results to the problem of finding finite blocklength bounds on the error probabilities occurring in quantum state parameter estimation.\\\\ 
\noindent
It would be interesting to address the following open questions: Firstly, the classical concentrations of Lipschitz functions on metric probability spaces have been shown to be equivalent to some concentration of measure inequalities (see e.g.~Theorem 3.4.1 of \cite{[RS14]}). In \cite{OW09}, Osborne and Winter derived a quantum analogue of a particular concentration of measure inequality, namely the Talagrand concentration inequality {\cite{TL13}, employing the notion of a quantum Hamming distance. Since our quantum concentration results are quantum extensions of concentration for Lipschitz functions, it would be interesting to see how they relate to theirs. Moreover, we expect our results to be generalizable to infinite dimensions and, for example, provide concentration of quantum Gaussian thermal states, thus extending the well-known concentration of Gaussian random vectors. Finally, it would be interesting to find applications of our quantum concentration results in the fields of quantum information theory and quantum computing.

\paragraph{Acknowledgements}
We would like to thank Daniel Stilck Franca as well as Sathyawageeswar Subramanian for helpful discussions.
\bibliographystyle{abbrv}
\bibliography{libraryone}
\appendix
\end{document}